\documentclass[journal,twocolumn]{IEEEtran}

\usepackage{palatino}
\usepackage{mathtools}
\usepackage{url}
\usepackage{amsmath}
\usepackage{amssymb}
\usepackage{wrapfig}
\usepackage{framed}
\usepackage{mdframed}
\usepackage{verbatim}
 \usepackage{amsthm}
 \usepackage{cite}
\usepackage{subfigure}
\usepackage[table]{xcolor}
\usepackage{palatino}
\usepackage{mathtools}
\usepackage{url}
\usepackage{amsmath}
\usepackage{amssymb}
\usepackage{wrapfig}
\usepackage{framed}
\usepackage{mdframed}
\usepackage{bbm}
\usepackage{verbatim}\usepackage{subfigure}
\usepackage[utf8]{inputenc}
\usepackage{amssymb}
\usepackage{amsthm}
\usepackage{dsfont}

\usepackage{bbm}
\usepackage{subfigure}
\usepackage{times}
\usepackage{multirow}
\usepackage[final]{graphicx}
\usepackage{upgreek}
\usepackage{latexsym,amssymb,mathrsfs}
\usepackage{comment}
\usepackage{url}

\usepackage{algorithmicx}
\usepackage[ruled]{algorithm}
\usepackage{algpseudocode}
\usepackage{algpascal}
\usepackage{algc}
\usepackage{xcolor}
\usepackage{cite}

 \usepackage{cite}
\newtheorem{thm}{Theorem}

\newtheorem{prop}{Proposition}

\newtheorem{cor}{Corollary}
\newtheorem{rmk}{Remark}

\title{Answering Count Queries for Genomic Data with Perfect Privacy}
\author{\IEEEauthorblockN{Bo~Jiang \quad Mohamed Seif \quad Ravi~Tandon \quad Ming~Li\\
}
\IEEEauthorblockA{Department of Electrical and Computer Engineering}\\
\IEEEauthorblockA{University of Arizona, Tucson, AZ, 85721\\
Email: $\{\textit{bjiang, mseif, tandonr, lim}\}$@email.arizona.edu}
}
\begin{document}
\maketitle
\begin{abstract}
In this paper, we consider the problem of answering count queries for genomic data subject to perfect privacy constraints. 
Count queries are often used in applications that collect aggregate (population-wide) information from biomedical Databases (DBs) for analysis, such as Genome-wide association studies. Our goal is to design mechanisms for answering count queries of the following form: \textit{How many users in the database have a specific set of genotypes at certain locations in their genome?} At the same time, we aim to achieve perfect privacy (zero information leakage) of the sensitive genotypes at a pre-specified set of secret locations. The sensitive genotypes could indicate rare diseases and/or other health traits one may want to keep private. We present both local and central count-query mechanisms for the above problem that achieves perfect information-theoretic privacy for sensitive genotypes while minimizing the expected absolute error (or per-user error probability, depending on the setting) of the query answer. We also derived a lower bound of the per-user probability of error for an arbitrary query-answering mechanism that satisfies perfect privacy. We show that our mechanisms achieve error close to the lower bound, and match the lower bound for some special cases. We numerically show that the performance of each mechanism depends on the data prior distribution, the intersection between the queried and sensitive 
genotypes, and the strength of the correlation in the genomic data sequence.
\end{abstract}

\section{Introduction} \label{sec:introduction}

Genetic research is experiencing an unprecedented growth in terms of the amount of personal data collected in large repositories \cite{ohno2017finding}. The human genome is the complete set of genetic information, which is composed of four different bases (A, C, G, T). Genetic information is encoded inside chromosomes, and each chromosome contains genes responsible for various functions controlling the human body all together\cite{AKGUN2015103}. They provide medical researchers invaluable information that can play a role in different interesting applications such as sequencing \cite{motahari2013information}, genome sequence assembly \cite{shomorony2016information, si2017information} and Genome-Wide Association Studies (GWAS) \cite{cho2018secure}.
The main objective of these applications is to enable the investigation of variants in human genomes with different biological or health-related traits through interactive queries with biomedical databases (DBs) or directly from each record. For example, medical researchers might want to query how many users in a dataset have specific genotypes at certain locations in their genome. This can help to reveal relationships that exist in certain genotypes and biological traits or discover particular diseases like cancer \cite{easton2008genome}, diabetes, and even more complex ones \cite{hirschhorn2005genome}. Such systems can also be used as a powerful tool for finding new drug targets \cite{cao2014gwas}.

\begin{figure*}[t]
    \centering
    {\includegraphics[width=1.5\columnwidth]{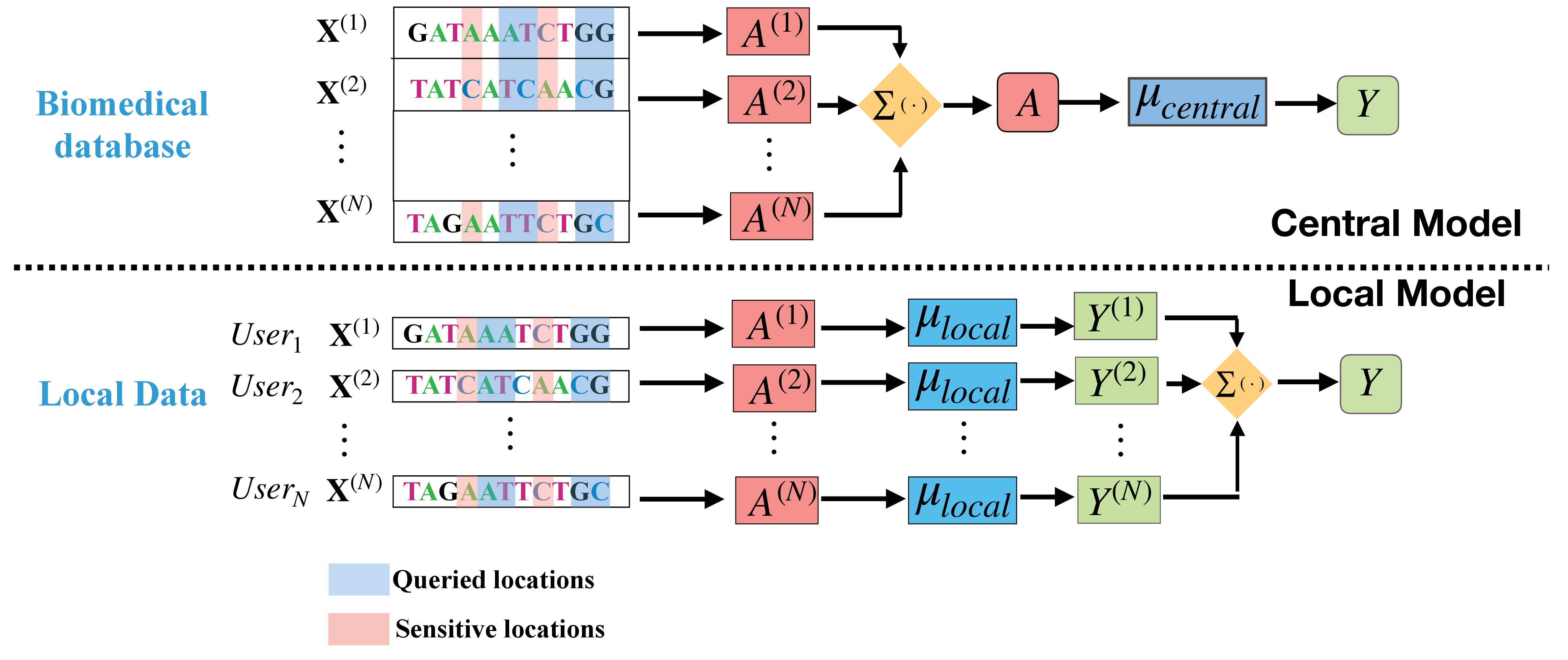}
    \caption{The system model. The queried sub-sequence $X_{\mathcal{L}}^{(m)}$ for the $m$-th user  and the query $v_{\mathcal{L}}$ altogether pass through a deterministic equality test function with a true answer $A^{(m)}$, followed by a private release mechanism. 
    The perturbed answer $Y$ guarantees perfect privacy for sensitive genotypes.
    }
    \vspace{-10pt}
    \label{fig:genome_system_model}}
\end{figure*}

Despite the impact of these applications on health services,  privacy remains a fundamental concern that needs to be addressed. Indeed, existing query-answering systems such as STRIDE \cite{lowe2009stride} and i2b2 \cite{murphy2011strategies} can leak sensitive genotypes about individuals \cite{homer2008resolving, shringarpure2015privacy, deznabi2017inference}. Therefore, a significant research effort has been made to enable privacy-preserving access to genomic data. Recent interesting works in the literature addressed this issue in genomics under different settings, including sequential genomic data release \cite{ye2020mechanisms}, query-answering in biomedical DBs \cite{raisaro2017addressing, simmons2016enabling, cho2020privacy}. Depending on whether the mechanism relies on a trusted third party, these privacy protection models can be classified into two categories: centralized models assume there is a trusted server that processes the dataset and releases privatized answers to specific types of queries. For example, Cho et al. \cite{cho2020privacy} proposed differentially private (DP) mechanisms for count queries using geometric additive noise mechanisms. Local models, on the other hand, enables each data holder (user) to perturb his/her genomic data sequence locally before publishing it to the server, which can be used to answer any query afterward, such as \cite{ye2020mechanisms}. 

In the past two decades, DP has been widely accepted as the de facto standard in the privacy research community. However, DP only leads to a few real-world adoptions. One of the main reasons is that DP is a stringent (worst-case) privacy notion, i.e., there is no assumption on the underlying prior distribution of the data. While privacy guarantees hold for the worst-case realizations of the data, this could lead to lower data utility. Hence, it is not ideal for scenarios when the prior distribution of genomic data is available, e.g., through public datasets and/or via existing statistical models of genomic data sequence\cite{li2003modeling,jiang2021context}. In contrast, context-aware privacy notions, such as mutual information privacy, can lead to mechanism designs with a better privacy-utility tradeoff by leveraging the prior information. Note that despite the high diversity, most of the DNA Sequence is common across the whole human population. Only around $0.5\%$ of each person's DNA is different from the reference genome, owing to genetic variations\cite{venter2001sequence}. This makes it suitable to adopt context-aware privacy notions for genomic sequence. Besides, context-aware privacy notions model each genomic data sequence as a random variable, which considers the correlation inside the data sequence.

In this paper, we consider  $K$ individuals participating in a genome variants survey, and each person's $N$-length genome sequence is drawn from a known distribution. Specifically, our goal is to design mechanisms for answering count queries of the following form: \textit{How many users have a specific set of genotypes (denoted by a set $v_{\mathcal{L}}$) at certain locations (denoted by $\mathcal{L}$) in their genome?} At the same time, we aim to achieve perfect secrecy (zero information leakage)  (perfect information-theoretic privacy) about sensitive genotypes at a pre-specified set $\mathcal{S}$ of secret locations. Our main contributions are summarized as follows:

\textbf{Main Contributions:}  

1) We present two count-query mechanisms for genomic data sequences in the central and local settings, respectively. The mechanisms achieve perfect privacy for sensitive genotypes at the locations belonging to a predefined set $\mathcal{S}$  while optimizing a utility metric (i.e., minimizing the probability of error for local setting and expected absolute distance for central setting). The proposed mechanisms incur less computation complexity and achieve significantly higher utility than a previous perfect secret mechanism \cite{ye2020mechanisms} and achieve a better utility-privacy tradeoff than DP-based mechanisms\cite{cho2020privacy}. 

2) We theoretically analyze the performance of each mechanism under different cases regarding the strength of correlations within the genomic sequence and the length of the overlapped genotypes between the queried locations and sensitive locations. The proposed mechanisms only depend on the statistics of the subsequences derived from the intersected set of locations. Thus, the complexity of our mechanisms only grows with $|\mathcal{L} \cup \mathcal{S}|$, irrespective of $N$. We also derive a lower bound on the error probability ($P_e$) for an arbitrary mechanism subject to perfect privacy constraints. We show the optimality of our proposed mechanisms for several special cases. 

3) We present comprehensive numerical simulations to show the utility-privacy tradeoff of the proposed mechanisms and compare them with existing DP-based mechanisms. Results show that our mechanisms achieve perfect privacy and incur smaller errors than those based on DP under certain large privacy budgets $\epsilon$ (larger $\epsilon$ indicates larger privacy leakage). We also show that the performance of the mechanisms is largely dependent on the data prior as well as the data dependence: the performance of each mechanism can be improved when the queried locations are less correlated with the sensitive locations. 

\section{Related Works}
A fundamental need in designing genomic privacy-preserving mechanisms is to selectively limit the leakage of information about biological
or health-related traits of an individual that can be inferred from the shared genetic data. Under this context, in  \cite{bioinformati}, Thenen et al. have shown that beacons, previously deemed secure against re-identification attacks,  were vulnerable despite their stringent policy, due to the correlation and inter-dependence inside each genomic sequence. To this end, \cite{SIMMONS201654, cho2020privacy} study differential
privacy mechanisms for sharing aggregate genomic data. However, DP-based mechanisms provide worst-case privacy guarantees that incur too much noise and lead to decreased utility in terms of accuracy. Furthermore, traditional DP-based mechanisms do not leverage correlations that exist in the genomic sequence, which may degrade the privacy guarantees offered by DP. In \cite{liu2016dependence}, Nour et al. first demonstrate this drawback of DP by proposing an attack and then propose a mechanism for privacy-preserving sharing of statistics from genomic datasets to attain privacy guarantees while taking into consideration the data dependence. From there, we assert that information theoretical privacy measurements are more favorable for genomic data processing for the following reasons. 1. Different individuals' genomic sequences possess extremely similar patterns at certain locations, which provides a solid foundation for estimating the underlying data distribution. 2. Dependence and inter-correlations that exist in the genomic sequence can be explicitly measured and leveraged in the mechanism design. We next discuss some of the related works that have considered information-theoretic privacy within this context.

A privacy-preserving genomic data sequence release problem was studied in \cite{ye2020mechanisms}, where the genotypes at some specific sensitive locations and some non-sensitive locations but are correlated to sensitive locations are hidden. This model also provides a baseline approach for query answer release (perturb then query). However, the query-answering mechanism incurs much less noise than sequence release mechanisms  (either local or centralized ones). Other than that, the sequence release model was built for the local settings only. In this work, we study query-answering models in both local and central settings.

There is also a line of work on the privacy-preserving data release problem that is similar to our settings\cite{7918623,8262832, 8438536}. {The variables form a Markov chain $S\to X\to Y$, where $X$ and $Y$ represent the input and output data of the mechanism, respectively, and $S$ denotes a private latent variable that correlates to $X$}. Based on this model, \cite{7918623} proposes principal inertia components, which provide a fine-grained decomposition of the dependence between two random variables. It also proves that the smallest PIC of $P_{S,X}$ plays a central role in achieving perfect privacy (i.e. $I(S; Y) = 0$):
If $|X|\le|S|$, then perfect privacy is achievable with $I(X; Y) > 0$ if and only if the smallest PIC of $p_{S,X}$ is $0$. However, the mechanisms provided in \cite{7918623} are challenging to be implemented in genomic data, as it requires an exhausting search of functions in the defined region.

In \cite{8262832}, the system allows post-processing after taking observation on $Y$, then data utility is measured by the mean square error (MSE) between $X$ and the estimation $E[X|Y]$. Finally, the utility and privacy tradeoff can be formulated by minimizing the MSE while subject to certain privacy constraints. However, this work focuses on improving the utility-privacy tradeoff and do not provide closed-form privatization mechanisms. Also, this paper studies general privacy guarantees, not perfect privacy. 

In \cite{8438536}, privacy and utility are defined as the probabilities of correctly guessing $S$ and $X$ given $Y$, respectively. This paper also mentioned mechanism design. However, firstly, the perfect privacy measured by correct guessing cannot imply independence of the random variable (what we defined as perfect privacy in this paper), as $P_{C}(S|Y) = P_{C}(S)$ does not imply $P(S|Y) = P(S)$. Also, the modeling of the correlation between the latent variable and the input data is too simplified: the correlated latent variable considered in \cite{8438536} is also binary, and the correlation can be represented by two parameters. Finally, the mechanism is derived as a function of the correlation parameters. But for genomic data processing, the data sequence is a vector with each value having a cardinality of $4$. Thus the mechanism in \cite{8438536} does not apply to our problem. On the other hand, the model we considered is in the form of $S\to f(X) \to Y$, where $S$ and $X$ are vectors, $f(X)$ and $Y$ are binary, which is similar to, but different from related works in \cite{7918623,8262832, 8438536}.

\section{Model setup \& problem statement}

\begin{table}[t] 
\centering{
\caption{\small{List of symbols}}   
\begin{tabular}{c{l} c{l}}
\hline
$\mathbf{X}$ & Genomic data sequence \\
$N$ & Length of the genomic sequence\\
$k$ & individual index \\ 
$\mathcal{X}$ & Finite set \{A,T,G,C\}\\
${P(\cdot)}$ & Probability distribution \\
$\mathcal{S}$ & Set of indices of the sensitive locations\\
$\mathcal{Q}$ & Count query \\
$\mathcal{L}$ & Set of indices of the query locations\\
$v_{\mathcal{L}}$ & Reference sequence \\ 
$\mathcal{M}$ & Privacy protection mechanism \\ 
$Y$ & Release of the privacy protection mechanism\\
$P_e$ & Probability of error\\
$R$ & Dependence measurement\\
%$P_{\min}$ & $\min_{x\in{\mathcal{X}}}P_X(x)$\\
\hline
\end{tabular} }
\label{table:list_symbols}
\end{table}

Consider a set of $K$ sequences $\{\bold{X}^{(1)},\bold{X}^{(2)},...,\bold{X}^{(K)}\}$, where each sequence  $\bold{X}^{(k)}=\{X^{(k)}_1,X^{(k)}_2,...,X^{(k)}_N\}$,  {for $k=1,2,...,K$ is of length $N$}. The entries $X^{(k)}_n$ are drawn from a finite alphabet $\mathcal{X}$ (e.g. for genomic data, $\mathcal{X}=\{A,T,G,C\}$). % represents the user's index in the dataset and $n\in\{1,2,...,N\}$ stands for the index of the data in each individual's genome sequence. 
We assume that the sequences are independent and each sequence is drawn from a known distribution  ${P}_{\bold{X}}(\bar{x})={P}_{X_1,X_2,...,X_N}(x_1,x_2,...,x_N)$ (${P}_{\bold{X}}(\bar{x})$ can be different from user to user, we simplify by removing users' index). In users' genomic sequence, some locations are private and need to be protected while others are nonsensitive. Denote $\mathcal{S}=\{s_1,s_2,...,s_{|\mathcal{S}|}\}$ as the set of indices of the sensitive locations of the genomic data sequence (each $X^{(k)}_{n}$ is sensitive $\forall{k\in\{1,2,...,K\}, \forall{n}\in \mathcal{S}}$), and let $\bold{X}^{(k)}_{\mathcal{S}}=({X}^{(k)}_{s_1},{X}^{(k)}_{s_2},...,{X}^{(k)}_{s_{|\mathcal{S}|}})$ denote the sensitive genomic subsequence of the $k$-th user. 

A count query can be defined as follows: the number of users in the dataset with a specific sequence at some particular locations. Mathematically, we denote the query as $\mathcal{Q}=\{\mathcal{L},v_{\mathcal{L}}\}$, where $ \mathcal{L} =\{l_1,l_2,...,l_{|\mathcal{L}|}\}$ denotes the indices of the queried locations and $\bold{X}^{(k)}_{\mathcal{L}}=({X}^{(k)}_{l_1},{X}^{(k)}_{l_2},...,{X}^{(k)}_{l_{|\mathcal{L}|}})$ denotes the corresponding queried sequence for the $k$-th user; $v_{\mathcal{L}}$ denotes the reference sequence and $v_j$ denotes the reference value at location $j$. Then the local query of the $k$-th user is $Q^{(k)}=\{\bold{X}^{(k)}_{\mathcal{L}}=v_{\mathcal{L}}\}$. Note that, the sensitive locations $\mathcal{S}$ and the queried locations $\mathcal{L}$ may or may not have common indices. Denote $\bar{\mathcal{L}} = \mathcal{L} \backslash (\mathcal{L} \cap \mathcal{S})$ as the non-sensitive query locations, and $\bar{\mathcal{S}} = \mathcal{S} \backslash (\mathcal{L} \cap \mathcal{S})$.

Let $A^{(k)}$ be the true query answer calculated from the $k$-th user. Then $A^{(k)} =\mathbbm{1}_{\{\bold{X}^{(k)}_{\mathcal{L}} =v_{\mathcal{L}}\}}=\prod_{j=1}^{|\mathcal{L}|}\mathbbm{1}_{\{\bold{X}^{(k)}_{l_j} =v_j\}}$ and the final count can be expressed as $A=\sum_{k=1}^KA^{(k)}$, i.e., the aggregated answer is decomposable. To avoid leaking sensitive information about each $\bold{X}^{(k)}_{\mathcal{S}}$, we design private mechanisms and release a perturbed version while achieving high utility. Depending on whether there exists a trusted server, there are two different settings for the data protection mechanisms: for the local setting, each user privatizes $A^{(k)}$ independently and publishes a perturbed version $Y^{(k)}$, then $Y=\sum_{k=1}^KY^{(k)}$ is the aggregated result; For the central setting, the data server first aggregates $A$ from each $A^{(k)}$, then perturbs $A$ and releases $Y$ as a privatized version. Each of the local and central mechanisms has its pros and cons. Generally speaking, the central model guarantees a comparable level of privacy protection as the local model with a smaller amount of noise. The local model, on the other hand, provides a customizable privacy budget for each individual without relying on any third party. The system models considered are depicted in Fig. \ref{fig:genome_system_model}.

 {
\textbf{Data Privacy:} In  this paper, we aim to satisfy perfect privacy of sensitive genotypes, which is defined as:
\begin{equation}\label{PP_l}
    P_{Y^{(k)}|\bold{X}^{(k)}_{\mathcal{S}}}\left(y|\bold{x}_{\mathcal{S}}\right)=P_{Y^{(k)}}(y), \forall {k\in\{1,2,...,K\}},
\end{equation}
for local model, and 
\begin{equation}\label{PP_l2}
P_{Y|\bold{X}^{(k)}_{\mathcal{S}}}\left(y|\bold{x}_{\mathcal{S}}\right)=P_Y(y), \forall {k\in\{1,2,...,K\}}
\end{equation}}
 {for central model. The privacy-preserving mechanism must guarantee that the output of the privacy protection mechanism ($Y$ for central model, $Y^{(k)}$ for the local model) is independent of the genomic data $\bold{X}^{(k)}_{\mathcal{S}}$ at sensitive locations  $\mathcal{S}$ for every user $k$}. In the following, we refer to the conditions in \eqref{PP_l} as perfect privacy constraints. 

\textbf{Performance (Utility):} 
The performance of the mechanism is measured by the Expected Average Error (EAE) between $A$ and $Y$. Mathematically,
\begin{equation} \label{utility}
    \text{EAE}=E[|A-Y|].
\end{equation}
In this paper, we aim to maximize the utility (minimize the EAE) subject to perfect privacy.  {We summarized symbols used throughout this paper in Table \ref{table:list_symbols}.} 

%In the following sections, we will relate the measurement of local utility ($P_e$) and the aggregate utility. 
\section{Main Results}
In this paper, we aim to minimize the EAE subject to perfect local privacy. In this Section, we investigate privacy-preserving query answer mechanisms for local and central settings, respectively.  %We also present two illustrative examples that highlight the main ideas behind the mechanisms. We then present an information-theoretic lower bound on $P_e^{(k)}$ and show the optimality of our mechanisms for some special cases.
\subsection{Local Mechanism Design}
Next, we present two mechanisms based on different data processing settings which satisfy perfect local privacy for sensitive genotypes. To measure the performance of the proposed mechanisms for aggregated query, we define the utility for the aggregated query as the Expected Absolute Error (EAE):
\begin{equation}
    E[|Y-A|]={E}\left[\left|\sum_{k=1}^K \big(Y^{(k)} - A^{(k)}\big) \right|\right].
\end{equation}
 {
Note that the EAE can be further upper-bounded as follows:
\begin{small}
\begin{align}
    {E}[|Y-A|] 
    & \leq{\sum_{k=1}^K {E}\left[|Y^{(k)} - A^{(k)}|\right]} \triangleq \sum_{k = 1}^{K}P_{e}^{(k)}, \nonumber%\label{eqn:sum_individual_bounds}
\end{align}
\end{small}
where $P_e^{(k)}$ denotes the per-user error probability. }Given a query, we will design binary perturbation mechanisms for each sequence to minimize the local query error while satisfying per-user perfect privacy constraints. Formally, the problem can be described as 
    \begin{align}
    \min P_{e}^{(k)},~ \text{s.t.}~ \eqref{PP_l},~\forall{k=1,2,...,K}. \label{eqn:perfect_secrecy_condition}
    \end{align}
    
A local privacy-protection mechanism $\mathcal{M}^{(k)}$ with parameter $\mu(\bold{X}^{(k)},\mathcal{Q})$ can be described as follows:  $\mathcal{M}^{(k)}$ takes as input the local genomic sequence $\bold{X}^{(k)}$ of a user and the query $\mathcal{Q}=\{\mathcal{L},v_{\mathcal{L}}\}$. Then, it extracts the sensitive data sequence $\bold{X}^{(k)}_{\mathcal{S}}$ and the queried data sequence  $\bold{X}^{(k)}_{\mathcal{L}}$ respectively. From there, $\mathcal{M}^{(k)}$ releases $Y^{(k)}=0$ or $Y^{(k)}=1$ according to the data prior distribution as well as the correlation between   $\bold{X}^{(k)}_{\mathcal{S}}$ and $\bold{X}^{(k)}_{\mathcal{L}}$. The corresponding release probabilities are 
defined as:
\begin{align}
   & Y^{(k)}=  \nonumber
    \begin{cases}1,~~~~~ \text{w.p.}~~~~~ \mu(\bold{X}^{(k)},\mathcal{Q}), \\
   0, ~~~~~\text{w.p.} ~~~~~1-\mu(\bold{X}^{(k)},\mathcal{Q}).
   \end{cases}
\end{align} 
We present two mechanisms, $\mathcal{M}^{(k)}_1$, and $\mathcal{M}^{(k)}_2$, with parameters  $\mu_{1}(\bold{X}^{(k)},\mathcal{Q})$ and $\mu_{2}(\bold{X}^{(k)},\mathcal{Q})$ given as follows: 
\begin{align}\label{eqn:proposed_mechanism_1}
& \mu_{1}(\bold{X}^{(k)}, \mathcal{Q}) \triangleq  P_{Y^{(k)}|\bold{X}^{(k)}_{\bar{\mathcal{L}}}, \bold{X}^{(k)}_{\mathcal{S}}}(1| x_{\bar{\mathcal{L}}}, x_{\mathcal{S}}) =  \nonumber \\
&     \begin{cases}\frac{\min_{w} P_{\bold{X}^{(k)}_{\bar{\mathcal{L}}}|\bold{X}^{(k)}_{\mathcal{S}}}(x_{\bar{\mathcal{L}}}| w)}{P_{\bold{X}^{(k)}_{\bar{\mathcal{L}}}|\bold{X}^{(k)}_{\mathcal{S}}}(x_{\bar{\mathcal{L}}}| x_{\mathcal{S}})}, & \text{if} \hspace{0.05in} x_{\bar{\mathcal{L}}} = v_{\bar{\mathcal{L}}}, \mathcal{E} \leq 1/2, \\
   0, & \text{otherwise},
   \end{cases}
\end{align} 
\begin{align} \label{eqn:proposed_mechanism_2}
~~~~~~& \mu_2(\bold{X}^{(k)}, \mathcal{Q})   \triangleq P_{Y^{(k)}|\bold{X}^{(k)}_{\bar{\mathcal{L}}}, X^{(k)}_{\mathcal{S}}}(1| x_{\bar{\mathcal{L}}}, x_{\mathcal{S}}) = \nonumber \\
       &    \begin{cases} 1,& \text{if} \hspace{0.05in} x_{\bar{\mathcal{L}}} = v_{\bar{\mathcal{L}}}, \mathcal{E} \leq 1/2, \\
       1- \frac{\min_{w} P_{\bold{X}^{(k)}_{\bar{\mathcal{L}}}|\bold{X}^{(k)}_{\mathcal{S}}}(x_{\bar{\mathcal{L}}}| w)}{P_{\bold{X}^{(k)}_{\bar{
       \mathcal{L}}}|\bold{X}^{(k)}_{\mathcal{S}}}(x_{\bar{\mathcal{L}}}| x_{\mathcal{S}})}, &  \text{otherwise},
       \end{cases}
\end{align}
where $\mathcal{E}$ is defined as $\mathcal{E} \triangleq \operatorname{Pr} (\bold{X}^{(k)}_{\mathcal{L} \cap \mathcal{S}} \neq v_{\mathcal{L} \cap \mathcal{S}}) $. Both of the proposed mechanisms use the following probability ratio:
\begin{equation}\label{eqr}
    R(\bold{X}^{(k)}_{\bar{\mathcal{L}}},\bold{X}^{(k)}_{\mathcal{S}})=\frac{\min_wP_{\bold{X}^{(k)}_{\bar{\mathcal{L}}}|\bold{X}^{(k)}_{\mathcal{S}}}(x_{\bar{\mathcal{L}}}| w)}{P_{\bold{X}^{(k)}_{\bar{\mathcal{L}}}|\bold{X}^{(k)}_{\mathcal{S}}}(x_{\bar{\mathcal{L}}}| x_{\mathcal{S}})},
\end{equation}
which can be (informally) used to measure the statistical dependence between $\bold{X}^{(k)}_{\bar{\mathcal{L}}}$ (genotypes in $\mathcal{L}$ which are not in the sensitive locations) and $\bold{X}^{(k)}_{\mathcal{S}}$ (genotypes in sensitive locations). In the following, we use $R$  for simplicity. Specifically, a value of $R=1$ clearly implies independence of $\bold{X}^{(k)}_{\bar{\mathcal{L}}}$ and $\bold{X}^{(k)}_{\mathcal{S}}$, whereas $R<1$ implies that $\bold{X}^{(k)}_{\bar{\mathcal{L}}}$ and $\bold{X}^{(k)}_{\mathcal{S}}$ are dependent, and the dependence grows as $R$ becomes small. 

Given a query $ Q = \{\mathcal{L}, v_{\mathcal{L}}\}$, mechanism $\mathcal{M}^{(k)}_{1}$ first looks at the sequence $x_{\bar{\mathcal{L}}}$ that
corresponds to the non-sensitive genotypes in the set $\bar{\mathcal{L}}$. If $x_{\bar{\mathcal{L}}} = v_{\bar{\mathcal{L}}}$ and $ \mathcal{E} < 1/2$, the mechanism releases $1$ with probability $R$, otherwise, it always releases $0$. Mechanism $\mathcal{M}^{(k)}_{2}$ works similarly as follows. It first looks at the sequence $x_{\bar{\mathcal{L}}}$, if $x_{\bar{\mathcal{L}}} = v_{\bar{\mathcal{L}}}$ and $\mathcal{E} < 1/2$, then $\mathcal{M}^{(k)}_2$ always releases $1$. Otherwise, it releases $1$ with probability $ 1- R$. Our first main result is stated in the following Theorem.

\begin{thm}\label{thm:privacy_mechanism_case_1}   Mechanisms $\mathcal{M}^{(k)}_{1}$ and $\mathcal{M}^{(k)}_{2}$ satisfy perfect privacy.
\end{thm}
\begin{proof}
We expand $P_{Y^{(k)}|\bold{X}^{(k)}_{\mathcal{S}}}(1| x_{\mathcal{S}} ) $ in terms of $x_{\bar{\mathcal{L}}}$ using total probability Theorem. This yields the following set of steps:
\begin{small}
\begin{align}
    & P_{Y^{(k)}|\bold{X}^{(k)}_{\mathcal{S}}}(1| x_{\mathcal{S}} )  = \sum_{x_{\bar{\mathcal{L}}}} P_{Y^{(k)}|\bold{X}^{(k)}_{\bar{\mathcal{L}}}, X^{(k)}_{\mathcal{S}}}(1|x_{\bar{\mathcal{L}}}, x_{\mathcal{S}}) P_{\bold{X}^{(k)}_{\bar{\mathcal{L}}}|X^{(k)}_{\mathcal{S}}}(x_{\bar{\mathcal{L}}}|x_{\mathcal{S}}) \nonumber \\ 
    & = P_{Y^{(k)}|\bold{X}^{(k)}_{\bar{\mathcal{L}}}, \bold{X}^{(k)}_{\mathcal{S}}}(1|v_{\bar{\mathcal{L}}}, x_{\mathcal{S}}) P_{\bold{X}^{(k)}_{\bar{\mathcal{L}}}|\bold{X}^{(k)}_{\mathcal{S}}}(v_{\bar{\mathcal{L}}}|x_{\mathcal{S}}) \nonumber \\
    & \hspace{0.2in} + \sum_{x_{\bar{\mathcal{L}}} \neq v_{\bar{\mathcal{L}}}} P_{Y^{(k)}|\bold{X}^{(k)}_{\bar{\mathcal{L}}}, \bold{X}^{(k)}_{\mathcal{S}}}(1|x_{\bar{\mathcal{L}}}, x_{\mathcal{S}}) P_{\bold{X}^{(k)}_{\bar{\mathcal{L}}}|\bold{X}^{(k)}_{\mathcal{S}}}(x_{\bar{\mathcal{L}}}|x_{\mathcal{S}}) \nonumber \\ 
    & = \mu^{(k)}_{1}(v_{\bar{\mathcal{L}}}, Q, S) P_{\bold{X}^{(k)}_{\bar{\mathcal{L}}}|X^{(k)}_{\mathcal{S}}}(v_{\bar{\mathcal{L}}}|x_{\mathcal{S}}) \nonumber \\
    & \hspace{0.2in}+ \sum_{x_{\bar{\mathcal{L}}} \neq v_{\bar{\mathcal{L}}}} \mu^{(k)}_{1}(x_{\bar{\mathcal{L}}}, Q, S)  P_{\bold{X}^{(k)}_{\bar{\mathcal{L}}}|X^{(k)}_{\mathcal{S}}}(x_{\bar{\mathcal{L}}}|x_{\mathcal{S}}) \nonumber \\ 
    & \overset{(a)} = \min_{w} P_{\bold{X}^{(k)}_{\bar{\mathcal{L}}}|X^{(k)}_{\mathcal{S}}}(v_{\bar{\mathcal{L}}}|x_{\mathcal{S}} = w), \nonumber 
\end{align}
\end{small}
where step (a) follows from the proposed mechanism in \eqref{eqn:proposed_mechanism_1}.  Observe that the conditional probability does not depend on the realizations $x_{\mathcal{S}}$. Then, it is straightforward to show that 
\begin{align}
    P_{Y^{(k)}|\bold{X}^{(k)}_{\mathcal{S}}}(1| x_{\mathcal{S}} ) =   P^{(k)}_{Y}(1).  \nonumber 
\end{align}
The calculation for $\mathcal{M}^{(k)}_2$ follows on similar lines, and hence both mechanisms satisfy perfect privacy for sensitive genomes.
\end{proof}
We next present our second result which characterizes the error probability of the proposed mechanisms. 
\begin{thm}\label{thm:Pe}
Define $P^{(k)}_{e,1}$ and $P^{(k)}_{e,2} $ as the per-user error probability for $\mathcal{M}^{(k)}_1$ and $\mathcal{M}^{(k)}_2$, respectively, then:

\begin{small}
\noindent - \textit{Case 1:} $\mathcal{E} \leq 1/2$:
\begin{align}
  &P^{(k)}_{e,1} 
   =P_{\bold{X}^{(k)}_{\mathcal{L}}}(v_{\mathcal{L}})  + (2 \mathcal{E} - 1) \min_{w} P_{\bold{X}^{(k)}_{\bar{\mathcal{L}}}|\bold{X}^{(k)}_{\mathcal{S}}}(v_{\bar{\mathcal{L}}}|  w),  \\
   &P^{(k)}_{e,2} 
    =  1 - P_{\bold{X}^{(k)}_{\mathcal{L}}}(v^{(k)}_{\mathcal{L}}) -  \sum_{x_{\bar{\mathcal{L}}} \neq v_{\bar{\mathcal{L}}} } \min_{w} P_{\bold{X}^{(k)}_{\bar{\mathcal{L}}}|\bold{X}^{(k)}_{\mathcal{S}}}(x_{\bar{\mathcal{L}}}| w). \nonumber 
\end{align}
\noindent - \textit{Case 2:} $\mathcal{E} > 1/2$:
\begin{align}
  &P^{(k)}_{e, 1} = P_{\bold{X}^{(k)}_{\mathcal{L}}}(v_{\mathcal{L}}), \nonumber \\
   &P^{(k)}_{e, 2} 
     = 1 - P_{\bold{X}^{(k)}_{\mathcal{L}}}(v_{\mathcal{L}}) - \sum_{x_{\bar{\mathcal{L}}} \neq v_{\bar{\mathcal{L}}}  } \min_{w} P_{\bold{X}^{(k)}_{\bar{\mathcal{L}}}|\bold{X}^{(k)}_{\mathcal{S}}}(x_{\bar{\mathcal{L}}} | w) \nonumber \\ 
    & \hspace{0.2in} - (2 \mathcal{E} - 1) \min_{w} P_{\bold{X}^{(k)}_{\bar{\mathcal{L}}}|\bold{X}^{(k)}_{\mathcal{S}}}(v_{\bar{\mathcal{L}}} | w). \nonumber  
\end{align}
\end{small}
\end{thm}

The proof is presented in the Appendix.

%\begin{figure*}[t]
%\centering 
%\subfigure [Case 1, when $\mathcal{S}\cap\mathcal{L}=\emptyset$]
%{ \includegraphics[width=8cm]{local_1.pdf} 
%-}~
%\subfigure[Case 2, when $\mathcal{S}\cap\mathcal{L}=X_2$]
%{ \includegraphics[width=8cm]{Local_2.pdf} } 
%\caption{Description for the illustrative example.} 
%\label{numerical} 
%\vspace{-10pt}
%\end{figure*}

\begin{figure*}[t]
    \centering
    {\includegraphics[width=1.8\columnwidth]{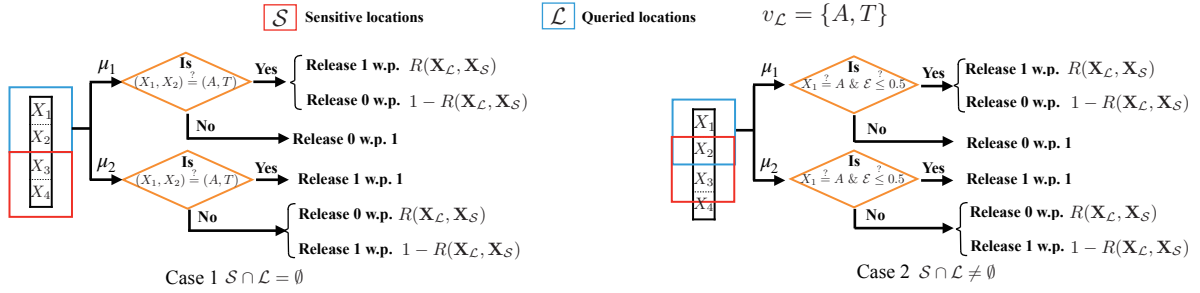}
    \caption{Description for the illustrative example for the two mechanisms $\mathcal{M}_1$ and $\mathcal{M}_2$.}
    \vspace{-10pt}
    \label{numerical}}
\end{figure*}

\subsubsection{Illustrative Example}
We next explain our mechanisms through an illustrative example. Consider a user's genome sequence  $\bold{X}=\{X_1,X_2,X_3, X_4\}$, and a target query  $\mathcal{Q}=\{\mathcal{L}=\{1,2\}$, $v_{\mathcal{L}}=\{A,T\}\}$, i.e., the query is of the form ``\textit{Is  $(X_1, X_2) = (A, T)$?}". In this example, we consider two scenarios for the set $\mathcal{S}$ of sensitive locations. Specifically, we can either have $\mathcal{S} \cap \mathcal{L} = \emptyset$ or  $\mathcal{S} \cap \mathcal{L} \neq \emptyset$.

%1) $ \bold{X}_{\mathcal{S}}=\{X_3,X_4\}$ and  2) $X_{\mathcal{S}}=\{X_2,X_3\}$. 

%value of $R=0$ implies  $X_{\bar{\mathcal{L}}}$ and $X_{\mathcal{S}}$ are highly dependent, whereas . 
%We consider three special values for this ratio: $R(X_{\mathcal{L}},X_{\mathcal{S}})\in \{0,1/2,1\}$, where $0$ means $X_{\bar{\mathcal{L}}}$ and $X_{\mathcal{S}}$ are highly correlated, $1$ means $X_{\bar{\mathcal{L}}}$ and $X_{\mathcal{S}}$ are independent, and $1/2$ denotes the strength of correlation is medium.

\noindent \underline{\textbf{Case 1:  $\mathcal{S} \cap \mathcal{L} = \emptyset$}}~~  When $\mathcal{S}=\{3,4\}$, if $R(\bold{X}_{\mathcal{L}},\bold{X}_{\mathcal{S}})=0$, mechanism $\mathcal{M}_1$ always releases $0$ and $\mathcal{M}_2$ always releases $1$. When $R(\bold{X}_{\mathcal{L}},\bold{X}_{\mathcal{S}})=1$, both mechanisms release the true answer, i.e., $Y=A$.  For the case when $0< R(\bold{X}_{\mathcal{L}},\bold{X}_{\mathcal{S}})< 1$, each mechanism perturbs the true answer as follows:

$\mathcal{M}_1$ first examines whether $\{X_1,X_2\}=\{A,T\}$, if no, then it releases the correct answer $Y=0$; if yes, then it releases $Y=1$ with probability $R(\bold{X}_{\mathcal{L}},\bold{X}_{\mathcal{S}})$, and $Y=0$ with the remaining probability $1- R(\bold{X}_{\mathcal{L}},\bold{X}_{\mathcal{S}})$. On the other hand, for mechanism $\mathcal{M}_2$, if $\{X_1,X_2\}=\{A,T\}$, then it releases $Y=1$; if not, then it releases $Y=0$ with probability $R(\bold{X}_{\mathcal{L}},\bold{X}_{\mathcal{S}})$ (and $Y=1$ with the remaining probability 1- $R(\bold{X}_{\mathcal{L}},\bold{X}_{\mathcal{S}})$). The intuition behind the mechanisms is to limit the leakage which arises due to the dependence between the genotypes in non-sensitive query locations $\bar{\mathcal{L}}$ and the sensitive locations $\mathcal{S}$. 

%uniformly sample $Y$ from $\{0,1\}$. For $\mathcal{M}_2$, if $\{X_1,X_2\}=\{A,T\}$, then release $Y=1$; if no, then uniformly sample $Y$ from $\{0,1\}$. Mathematically, 
%\begin{equation}
%\begin{aligned}
%&\mu_{1}(\bold{X}, \mathcal{Q}) \triangleq  P_{Y|X_1,X_2, X_3}(1| x_1, x_2,x_3)=\\
%   &\begin{cases} \frac{1}{2}, ~~~~~~~~~~~~~~~~~~~~~~~~ \text{if} \hspace{0.05in} \{X_1,X_2\} = \{A,T\}, \\
%   0,  ~~~~~~~~~~~~~~~~~~~~~~~~\text{if} \hspace{0.05in} \{X_1,X_2\} \neq \{A,T\}. 
%   \end{cases}\\
%&\mu_{2}(\bold{X}, \mathcal{Q}) \triangleq  P_{Y|X_1,X_2, X_3}(1| x_1, x_2,x_3)=\\
%   &\begin{cases}1, ~~~~~~~~~~~~~~~~~~~~~~~~ \text{if} \hspace{0.05in} \{X_1,X_2\} = \{A,T\}, \\
%   \frac{1}{2}, ~~~~~~~~~~~~~~~~~~~~~~~~ \text{if} \hspace{0.05in} \{X_1,X_2\} \neq \{A,T\}. 
%   \end{cases}   
%   \end{aligned}
%\end{equation}
%Moreover, when $\{X_1,X_2\}=\{A,T\}$, $\mathcal{M}_1$ has $1/2$ probability to answer correctly, while $\mathcal{M}_2$ always releases the true answer. On the other hand, when  $\{X_1,X_2\}\neq\{A,T\}$, $\mathcal{M}_1$ always releases the true answer, while $\mathcal{M}_2$ has $1/2$ probability to answer correctly. Therefore, the performance of each mechanism depends on the probability of $P_{X_1,X_2}(A,T).$

\noindent \underline{\textbf{Case 2:  $\mathcal{S} \cap \mathcal{L} \neq \emptyset$}}~~ We next consider the case when $\mathcal{S}= \{2,3\}$, i.e.,  $\mathcal{L}\cap \mathcal{S} = \{2\}$. Similar to Case 1, when $R(\bold{X}_{\mathcal{L}},\bold{X}_{\mathcal{S}})=0$, mechanism $\mathcal{M}_1$ always releases $0$ and $\mathcal{M}_2$ always releases $1$ (which clearly satisfies perfect privacy). When $ R(\bold{X}_{\mathcal{L}},\bold{X}_{\mathcal{S}})>0$, however, since the intersection $\mathcal{L}\cap \mathcal{S}= \{2\}$ is non-empty, both mechanisms first check if $\mathcal{E}= P(X_2 \neq v_2) = P(X_2 \neq T)<1/2$, and if yes, then proceed in a similar manner as before. Specifically, each mechanism first examines whether the non-sensitive query sequence $X_{\bar{\mathcal{L}}}$ matches $v_{\bar{\mathcal{L}}}$ and $\mathcal{E}<1/2$, if yes, then $\mathcal{M}_1$ releases $Y=1$ with probability $R(\bold{X}_{\mathcal{L}},\bold{X}_{\mathcal{S}})$, and $Y=0$ with the remaining probability $1- R(\bold{X}_{\mathcal{L}},\bold{X}_{\mathcal{S}})$; if no, it releases the correct answer $Y=0$. Both cases are illustrated in Fig. \ref{numerical}.

% uniformly sample $Y$ from $\{0,1\}$ and $\mathcal{M}_2$ releases $Y=1$; otherwise $\mathcal{M}_1$ releases $Y=0$ and $\mathcal{M}_2$ uniformly sample $Y$ from $\{0,1\}$, mathematically, each mechanism perturbs as follows:
% \begin{equation}
% \begin{aligned}
% &\mu_{1}(\bold{X}, \mathcal{Q}) \triangleq  P_{Y|X_1,X_2, X_3}(1| x_1, x_2,x_3)=\\
%   &\begin{cases}\frac{1}{2}, ~~~~~~~~~~~~~~~~ \text{if} \hspace{0.05in} X_1 = A, P_{X_2}(T)\le{1/2},\\
%   0,  ~~~~~~~~~~~~~~~~ \hspace{0.05in} \text{otherwise},\\
%   \end{cases}\\
% &\mu_{2}(\bold{X}, \mathcal{Q}) \triangleq  P_{Y|X_1,X_2, X_3}(1| x_1, x_2,x_3)=\\
%   &\begin{cases}1, ~~~~~~~~~~~~~~~~~~ \text{if} \hspace{0.05in} X_1 = A, P_{X_2}(T)\le{1/2}, \\
%   \frac{1}{2}, ~~~~~~~~~~~~~~~~~\text{otherwise}. \\
%   \end{cases}   
%   \end{aligned}
% \end{equation}
% Intuitively, the $Y$ must be independent of $X_2$ and $X_3$, but $P_e$ is dependent on $X_2$. Therefore, when $X_1=A$, the mechanism takes into account the probability of $X_2=T$, when $P_{X_2}(T)\ge{P(X_2\neq{T})}$, releasing $Y=1$ results in smaller probability of error, and vice versa.

%It is worth noting that our mechanism can also be used in the local setting where each user submits perturbed answers to the server for equality queries.

\begin{rmk}  As an alternative baseline scheme, one can perturb the whole genome sequence first by (for instance, by applying the scheme in \cite{ye2020mechanisms}) and then answer the count. However, perturbing the whole sequence is unnecessary, and the complexity of such a scheme will grow exponentially with the genome sequence length $N$. In contrast, the complexity of our schemes grows exponentially with $|\mathcal{L} \cup \mathcal{S}|$ which can be substantially smaller than $N$. 
\end{rmk}

\begin{cor}
We can readily specialize the general result of Theorem \ref{thm:Pe} for the case when genomic data sequences are i.i.d, and  each genotype is uniformly distributed over a finite alphabet $\mathcal{X}$. Define $\lambda=\left(\frac{1}{|\mathcal{X}|}\right)^{|\mathcal{L}|}$, then for $\mathcal{M}^{(k)}_1$ and $\mathcal{M}^{(k)}_2$, if $\mathcal{E}\le{0.5}$, $\min\{P^{(k)}_{e,1},P^{(k)}_{e,2}\}=0$; if $\mathcal{E}>{0.5}$, $\min\{P^{(k)}_{e,1},P^{(k)}_{e,2}\}=\lambda$, where $\mathcal{E}=1-P_{\bold{X}^{(k)}_{\mathcal{L}\cap\mathcal{S}}}(v_{\mathcal{L}\cap\mathcal{S}})=1-\left(\frac{1}{|\mathcal{X}|}\right)^{|\mathcal{L}\cap\mathcal{S}|}$. %The resulting EAE is therefore upper bounded by $K\lambda$. 
\end{cor}

\subsubsection{Lower Bound on $P_e$}
We next derive an information-theoretic lower bound on $P_e$ for any local mechanism that satisfies perfect privacy. The goal is to examine how far our mechanisms are from optimality.

\begin{thm}\label{thm:lowerbound}
For any privacy-preserving local data release mechanisms  that provide perfect privacy for genomic  data at sensitive locations $\mathcal{S}$, the $P_e$ is lower bounded as:
\begin{equation}
    P^{(k)}_e\ge h^{-1}\left\{{h(A^{(k)})-\min\{h(A^{(k)}_{\bar{\mathcal{L}}}|A^{(k)}_{{\mathcal{L}\cap\mathcal{S}}}),h(A^{(k)}|\bold{X}^{(k)}_{{\mathcal{S}}})\}}\right\},
\end{equation}
where $h(\cdot)$ denotes the binary entropy function,  $A^{(k)}_{\bar{\mathcal{L}}} \triangleq\mathbbm{1}_{\{\bold{X}^{(k)}_{\bar{\mathcal{L}}}=v_{\bar{\mathcal{L}}}\}}$ and $A^{(k)}_{{\mathcal{L}\cap\mathcal{S}}} \triangleq\mathbbm{1}_{\{\bold{X}^{(k)}_{\mathcal{L}\cap\mathcal{S}}=v_{\mathcal{L}\cap\mathcal{S}}\}}$.
\end{thm}
This bound follows by using Fano’s inequality and the privacy constraint in \eqref{PP_l}. Detailed proof is provided in the Appendix. Denote $LB(P^{(k)}_e)$ as the lower bound of $P^{(k)}_e$ calculated from Theorem \ref{thm:lowerbound}, we next consider some special cases and compare the lower bound with the error probability of the proposed mechanisms. 
\begin{rmk} Under the following scenarios, $LB(P^{(k)}_e)$ from Theorem \ref{thm:lowerbound} matches $\min\{P^{(k)}_{e,1}, P^{(k)}_{e,2}\}$ from Theorem \ref{thm:Pe}:

(1)  When $\mathcal{L}\cap{\mathcal{S}}=\emptyset$: we consider two scenarios regarding the data correlation between $\bold{X}^{(k)}_{\mathcal{L}}$ and $\bold{X}^{(k)}_{\mathcal{S}}$.
    
    a) If $I(A^{(k)};\bold{X}^{(k)}_{\mathcal{S}})=h(A^{(k)})$,  $LB(P^{(k)}_e)=\min\{P^{(k)}_{e,1},P^{(k)}_{e,2}\}=\min\{P^{(k)}_A(0),P^{(k)}_A(1)\},$ which means when  $\bold{X}^{(k)}_{\mathcal{S}}$ and $\bold{X}^{(k)}_{\mathcal{L}}$ are fully dependent, any mechanism cannot outperform directly sampling result from the prior distribution.

b) If $I(A^{(k)};\bold{X}^{(k)}_{\mathcal{S}})=0$, $LB(P^{(k)}_e)=\min\{P^{(k)}_{e,1},P^{(k)}_{e,2}\}=0$, which implies when $\bold{X}^{(k)}_{\mathcal{L}}$ does not depend on $\bold{X}^{(k)}_{\mathcal{S}}$, our mechanisms achieve $P^{(k)}_e=0$, and matches the lower bound.

(2) When $\mathcal{L}\in\mathcal{S}$: 
$LB(P^{(k)}_e)=\min\{P^{(k)}_{e,1},P^{(k)}_{e,2}\}=0$, which is a special case of $I(A^{(k)};\bold{X}^{(k)}_{\mathcal{S}})=h(A^{(k)})$.

\end{rmk} 

In Section IV, we compare the lower bound with the error of the proposed mechanisms via simulations.

\subsubsection{Utility of Private Aggregated Query}

%where $P_{e}^{(k)}= \text{Pr}(Y^{(k)} \neq A^{(k)})$ can be viewed as the mis-classification error of a local mechanism for user $k$'s genome sequence. Therefore, minimizing the upper bound on EAE leads naturally to a parallelizable problem, i.e., 

%As given a query realization, one can always pick the mechanism ($\mathcal{M}_{1}$ or $\mathcal{M}_{2}$) with the smaller error, the total EAE is upper bounded by,

Next, we show that even though the proposed mechanisms are designed to minimize $P_e$, under the i.i.d setting (users are independent), using the summation of $P_e$ as a measure of EAE does not lose any optimality.

\begin{prop}
For the proposed mechanisms, if the sequences across users are i.i.d, then the EAE is given as follows:
\begin{align} 
    {E}[|Y-A|] = \sum_{k = 1}^{K} \min \left[ P_{e,1}^{(k)}, P_{e, 2}^{(k)} \right]. \nonumber
\end{align}
\end{prop}
The proof is provided in the Appendix. 

\begin{figure*}[t]
\centering 
\subfigure[Length of intersected data is $0$] 
{ \includegraphics[width=0.30\textwidth]{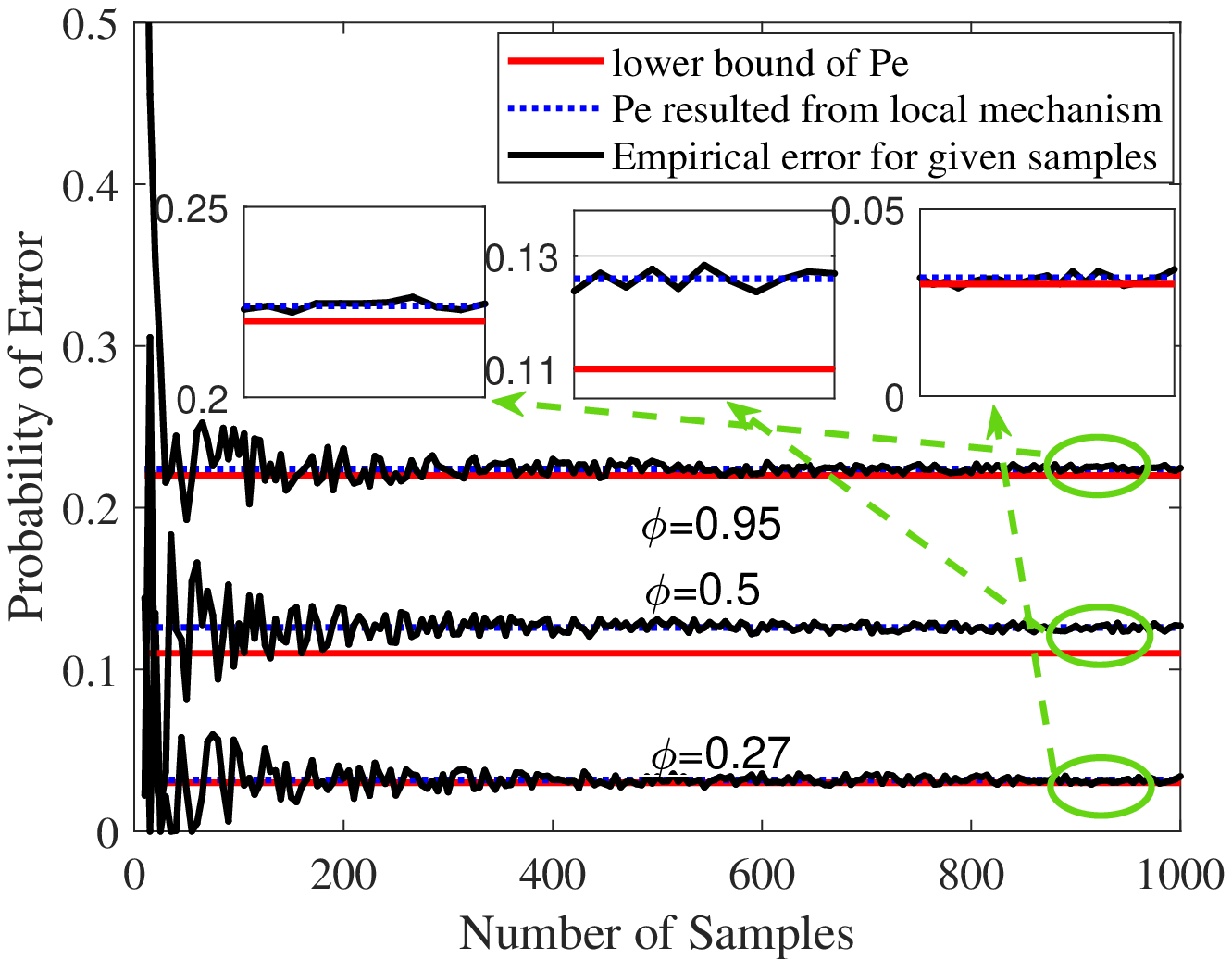}
\label{rs1} }
\subfigure[Length of intersected data is $4$] 
{ \includegraphics[width=0.3\textwidth]{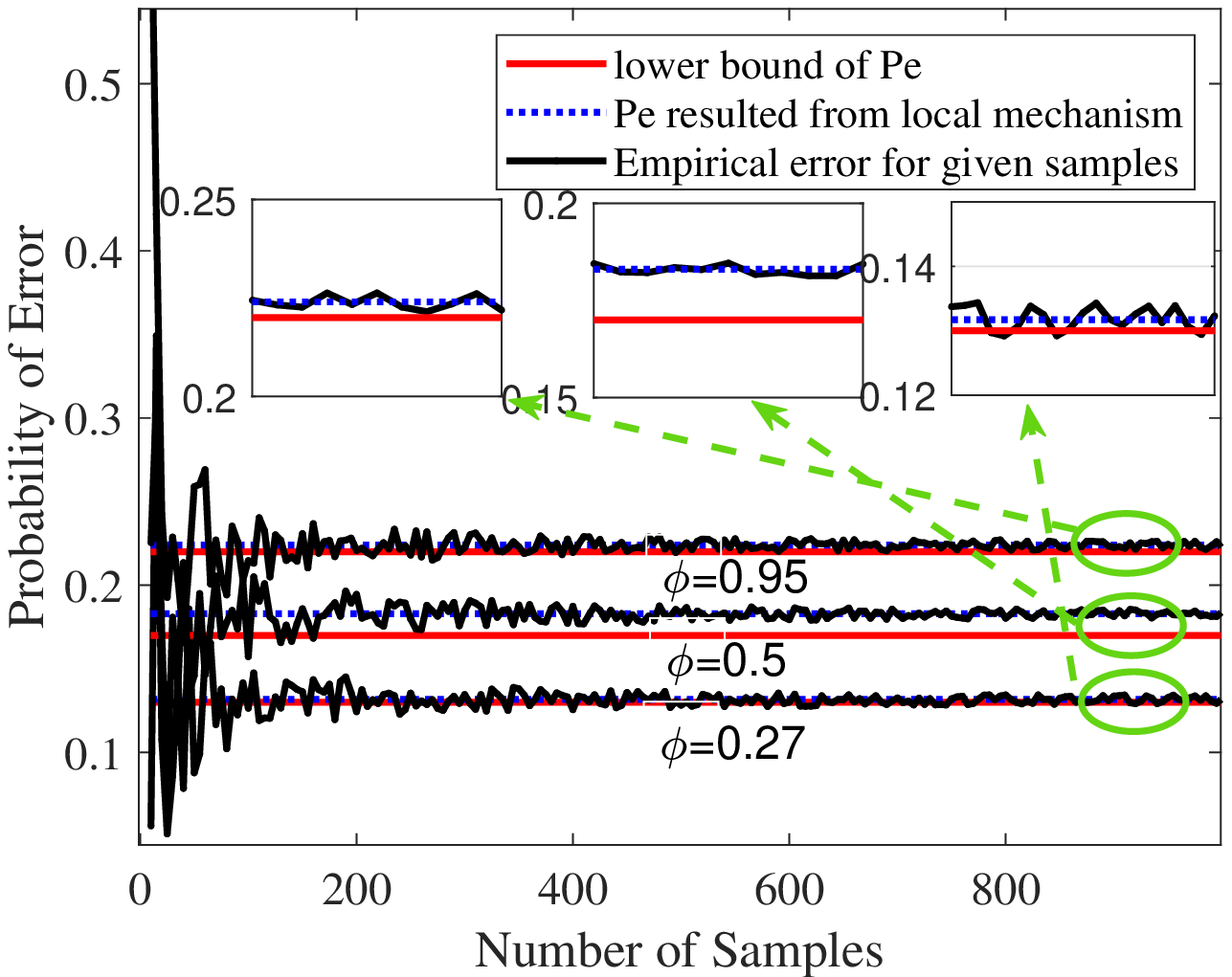}
\label{rs2} }
\subfigure[Length of intersected data is $8$] 
{ \includegraphics[width=0.30\textwidth]{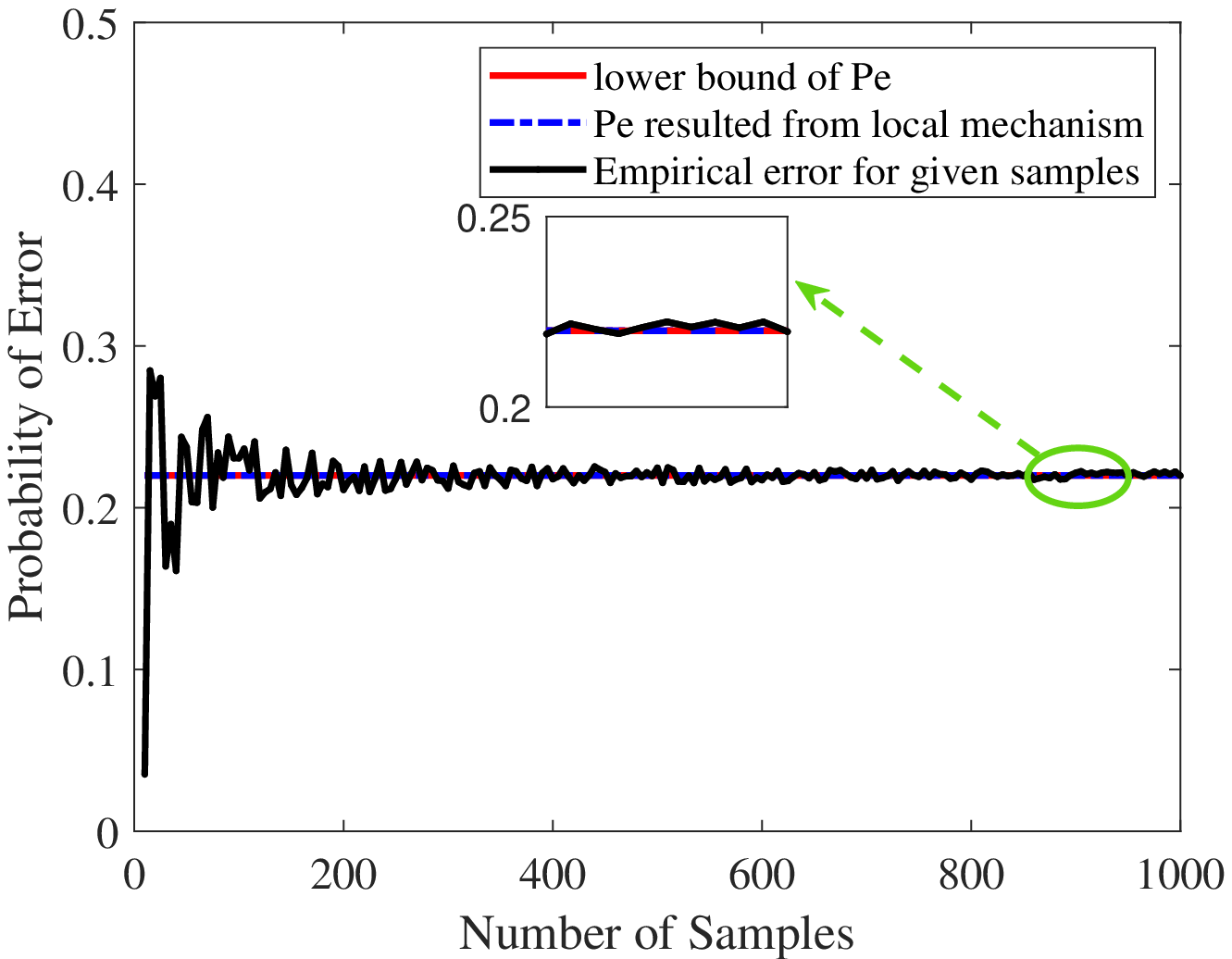}
\label{rs3} }
\caption{Comparison between $P_e$ and lower bound under different lengths of intersected data (correlation is measured by $\phi$)}
\label{lower} 
\vspace{-10pt}
\end{figure*}

\subsection{Centralized Mechanism}
In the following, we present the privacy-preserving mechanism in the centralized setting. In the centralized setting, the mechanism takes the whole dataset as input and then aggregates the true answer to some query $\mathcal{Q}=\{\mathcal{L}, v_{\mathcal{L}}\}$. The aggregated result is denoted by $A$. Finally, the mechanism perturbs $A$ and releases a privatized version $Y$ as the output.  {We denote $\bold{X}_{\mathcal{S}}=\{\bold{X}_{\mathcal{S}}^{(k)}\}^{(K)}_{k=1}$ as the sensitive data matrix}, which includes each user's genomic sequence at sensitive locations.

Problem formulation: The utility measurement in the centralized model is identical to that of \eqref{utility}. For the privacy constraints of the central model, the mechanism should provide the answer $Y$, which is independent of $\bold{X}_{\mathcal{S}}$.
It can be easily shown that, $Y$ is independent of each $\bold{X}^{(k)}_{\mathcal{S}}$ for $k\in \{1,..,K\}$. As such, the problem in the central setting can be represented as:
\begin{equation}
    \min E[|A-Y|] ~~ \text{s.t.} ~P_{Y|\bold{X}_{\mathcal{S}}}(y|\bold{x}_{\mathcal{S}}) = P_{Y}(y).
\end{equation}

The privacy-preserving mechanism in the central setting, $\mu_3$ can be specified as follows:
\begin{equation}\label{central_mec}
\begin{aligned}
 &\mu_{3}(D, \mathcal{Q})   \triangleq P_{Y|A,\bold{X}_{\mathcal{S}}}(y| a, \bold{x}_{\mathcal{S}})=\\
   &\begin{cases} P_A(y)+(1-P_A(y))\frac{\min_wP_{A|\bold{X}_{\mathcal{S}}}(a|w)}{P_{A|\bold{X}_{\mathcal{S}}}(a|\bold{x}_{\mathcal{S}})}, ~~~~~~~~~~~ \text{if} \hspace{0.05in} y = a, \\
   P_A(y)\left[1-\frac{\min_wP_{A|\bold{X}_{\mathcal{S}}}(a|w)}{P_{A|\bold{X}_{\mathcal{S}}}(a|\bold{x}_{\mathcal{S}})}\right],  ~~~~~~~~~~~~~~~~~~~~~\text{if} \hspace{0.05in} y \neq a. 
   \end{cases}\\
   \end{aligned}
\end{equation}
\begin{thm}
The mechanism provided in \eqref{central_mec} achieves perfect privacy. 
\end{thm}

Detailed proof is provided in Appendix. More related discussion to perfect privacy can be found in \cite{32120}. It is worth noting that in the central model, the term
\begin{equation}
    R = \frac{\min_wP_{A|\bold{X}_{\mathcal{S}}}(a|w)}{P_{A|\bold{X}_{\mathcal{S}}}(a|\bold{x}_{\mathcal{S}})}
\end{equation}
can be used implicitly to measure the dependence in the dataset. i.e., when the dependence among data is strong, there could be some $w$, such that $\min_wP_{A|\bold{X}_{\mathcal{S}}}(a|w) = 0$, and thus $R = 0$. On the other hand, when data are independent of each other $\min_wP_{A|\bold{X}_{\mathcal{S}}}(a|w)= P_{A|\bold{X}_{\mathcal{S}}}(a|\bold{x}_{\mathcal{S}}) = P_{A}(a)$, and hence $R = 1$. Therefore, $R\in[0,1]$ implicitly measures the dependence among the dataset. Consider the following three different cases regarding the mechanism in \eqref{central_mec}:
\begin{itemize}
    \item When $R=1$, data is independent of each other, the mechanism releases $y=a$ with a probability of $1$. i.e., it will always release the true aggregate.
    \item When $R=0$, data is highly dependent on each other, the mechanism releases $y=a$ with a probability of $P_A(a)$, and the probability of releasing any other $y\neq{a}$ is also $P_A(a)$, which means, the mechanism release an answer by sampling from the distribution of $A$.
    \item When $R\in(0,1)$, the mechanism releases $y=a$ with probability that is larger than the prior of $P_A(y)$. Note that the larger the probability is, the higher utility the mechanism achieves, and this accuracy depends on the data dependence.
\end{itemize}
\begin{thm}
The mechanism in \eqref{central_mec} incurs the following expected error.
\begin{equation}
    \sum_{a=0}^N\sum_{y\neq{a}}|y-a|P_A(y)\left[P_A(a)-\min_{w}P_{A|X_{\mathcal{S}}}(a|w)\right].
    \end{equation}
\end{thm}

Detailed proof is provided in the Appendix.
% Let's assume the dependence of $A$ and $\bold{X}_{\mathcal{S}}$ can be measured by the value of $\min_wP_{A|\bold{X}_{\mathcal{S}}}(a|w)$, which takes value of between $0$ and $P_{A}(a)$, where $0$ means strong dependence, and $P_A(a)$ otherwise. Then when the dependence is strong, the mechanism randomly sample from the distribution of $Y$ such that the distribution of $Y$ is identical to that of $A$. When the dependence is weak, the mechanism always releases the true answer. 

\textbf{Illustrative Example} Consider a dataset holding $3$ users' genomic data with each data uniformly distributed. Users' data are independent of each other. $\mathcal{L}=\{2\}$, $\mathcal{S}=\{1\}$ and $v=T$. Next, we show how the $\mu_{\text{central}}$ works.

$P_{X_{\mathcal{L}}}(v_{\mathcal{L}})   =P_{X_2}(T)=1/4$. Then $P_A(a)$ satisfies binomial distribution where
\begin{equation}
    P_A(a)=\left(\begin{array}{cc}   
    3\\ 
    a
  \end{array}\right)(1/4)^a(3/4)^{(n-a)}.
\end{equation}
The correlation term can be expressed as:
\begin{equation}\label{eq13}
    \begin{aligned}
    &P_{A|\bold{X}_{\mathcal{S}}}(a|\bold{x}_{\mathcal{S}})\\
    =&P_{A|\{\bold{X}^{(k)}_{1}\}_{k=1}^3}(a|\{\bold{x}^{(k)}_{1}\}_{k=1}^3)\\
    =&\sum_{\{x^{(k)}_{2}\}_{k=1}^3}P_{A|\{X^{k}_2\}_{k=1}^3}(a|\{x^{(k)}_{2}\}_{k=1}^3)\\
    &\cdot P_{\{X^{k}_2\}_{k=1}^3|\{X^{(k)}_1\}_{k=1}^3}(\{x^{k}_2\}_{k=1}^3|\{x^{k}_1\}_{k=1}^3)\\
    =&\sum_{\{x^{(k)}_{2}\}_{k=1}^3}P_{A|\{X^{k}_2\}_{k=1}^3}(a|\{x^{(k)}_{2}\}_{k=1}^3)\prod_{k=1}^3P_{X^{k}_2|X^{k}_1}(x^{k}_2|x^{k}_1),\\
    \end{aligned}
\end{equation}
where 
\begin{equation}
\begin{aligned}
&P_{A|\{X^{(k)}_2\}_{k=1}^3}(a|\{x^{(k)}_{2}\}_{k=1}^3)=\\
   &\begin{cases}1, ~~~~~~~~~~~~~~~~ \text{if} \hspace{0.05in} \sum_{k=1}^3\mathbbm{1}_{\{x^{(k)}_2=T\}}=a,\\
   0,  ~~~~~~~~~~~~~~~~ \text{otherwise},\\
   \end{cases}\\
\end{aligned}
\end{equation}
Then \eqref{eq13} can be writen as:
\begin{equation}
\sum_{\{x^{(k)}_{2}\}_{k=1}^3\in\mathcal{B}}\prod_{k=1}^3P_{X^{(k)}_2|X^{(k)}_1}(x^{(k)}_2|x^{(k)}_1),
\end{equation}
where $\mathcal{B}=\{\{x^{(k)}_2\}_{k=1}^3:~\sum_{k=1}^3\mathbbm{1}_{\{x^{(k)}_2=T\}}=a\}$.
%\begin{equation}
%    \begin{aligned}
%    \mathcal{E}=1-P_{X_{\mathcal{L}\cap\mathcal{S}}}(v_{\mathcal{L}\cap\mathcal{S}})=1-\left(\frac{1}{C_{x}}\right)^{|\mathcal{L}\cap\mathcal{S}|}.
%    \end{aligned}
%\end{equation}
%Then each error probability shown in Theorem 1 is shown in the following proposition:
Under the i.i.d model, the EAE of $\mu_{\text{central}}$ is shown in the following corollary.
\begin{cor}
Under the setting where each user's genomic data sequence is i.i.d, and the distribution of each genomic data is uniform, Let $\lambda=\left(\frac{1}{C_x}\right)^{|\mathcal{L}|}$, then, the EAE for the central mechanism can be expressed as: When $\mathcal{E}\le{0.5}$, $\text{EAE}=0$, when $\mathcal{E}>{0.5}$, 
    $\text{EAE}=2\{\sum_{a=0}^NaP_A(a)F_A(a)-N\lambda^2\}$.
\end{cor}
\section{Numerical Evaluation}

We next numerically evaluate the proposed mechanisms for two cases: for the first case, we consider a first-order Markov property in each of the genomic sequences; for the second case, we generate synthetic data with a hidden Markov model (described later in sub-section V-E).

We next summarize the Markov setting: It is assumed that each user's genomic data has a first-order Markov property\cite{ye2020mechanisms}, i.e.,  $P_{X_{j+1}|X_{j},\bold{X}^{j-1}_{1}}(x_{j+1}|x_j,\bold{x}^{j-1}_{1})=P_{X_{j+1}|X_{j}}(x_{j+1}|x_{j})$. Denote $\mathcal{T}$ as the transition matrix from time $j$ to $j-1$, for all $j\in[1,|X|]$, and can be specified as:
\begin{equation}
    \begin{aligned}
        &\text{Pr}(X_{j+1}=x|X_{j}=x)=\phi\\
        &\text{Pr}(X_{j+1}=x'|X_{j}=x)=(1-\phi)/3\\
    \end{aligned}
\end{equation}
where $x,x'\in\{A,T,G,C\}$ and $x\neq x'$.
For the following evaluations, we randomly generate the user's data according to the prior and the transition matrix, and we run Monte-Carlo simulations to get an average error. We consider $K=1000$ users in the system. Each user possesses a genomic sequence with a length of $10$.
%\begin{equation}
%    \left[                 
%  \begin{array}{cccc}   
%    \phi & (1-\phi)/3 & (1-\phi)/3 & (1-\phi)/3\\  
%    (1-\phi)/3 & \phi & (1-\phi)/3 & (1-\phi)/3\\
%    (1-\phi)/3 & (1-\phi)/3 & \phi &(1-\phi)/3\\
%    (1-\phi)/3 & (1-\phi)/3 & (1-\phi)/3 & \phi
%  \end{array}
%\right].       
%\end{equation}
\begin{figure*}[htp]
\centering  
\subfigure[$\mathcal{S}\cap\mathcal{L}=\emptyset$, $P_{X_1}$ is uniform]{ \includegraphics[width=0.23\textwidth]{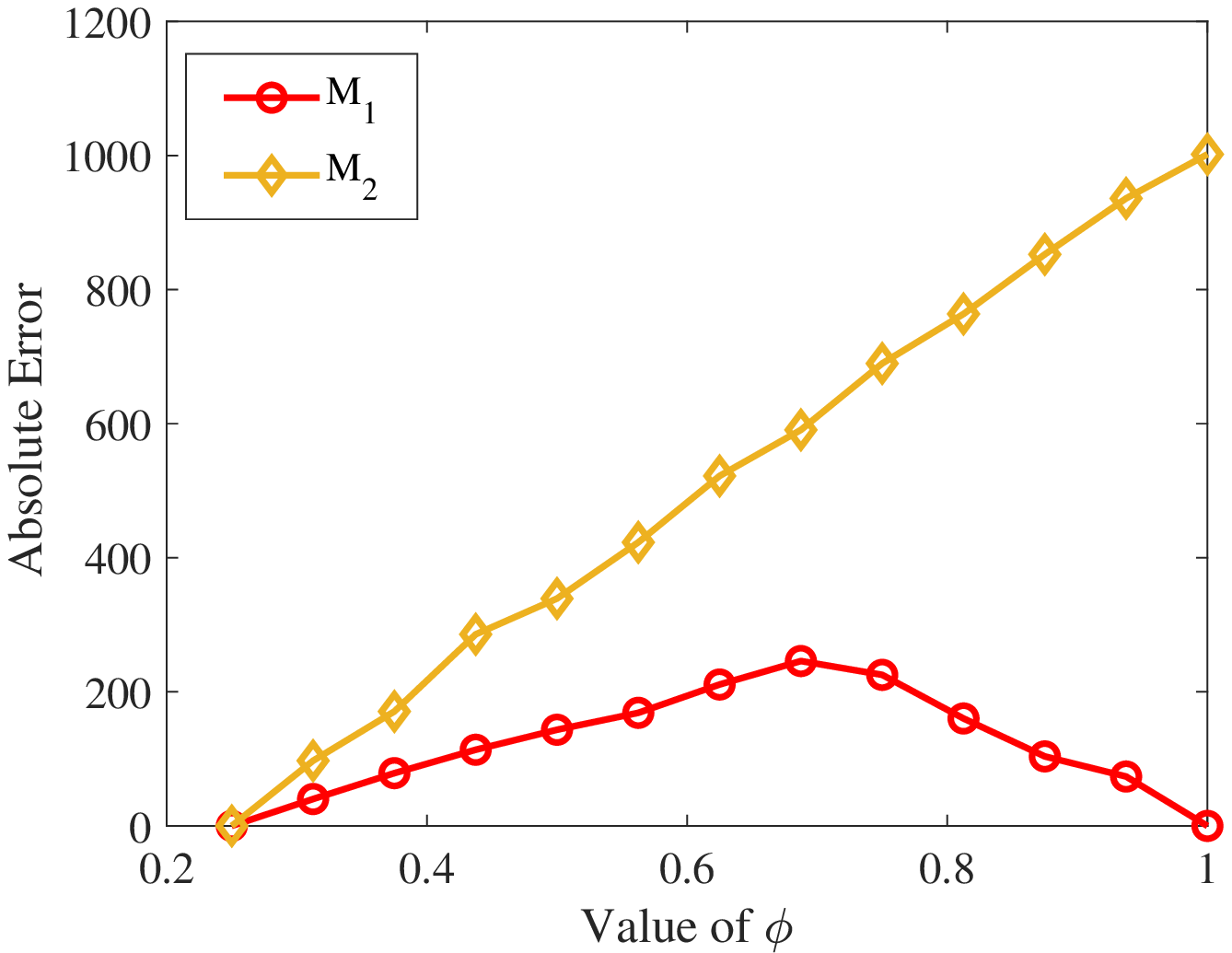} 
\label{case 1} }
\subfigure[When $\mathcal{S}\cap\mathcal{L}\neq\emptyset$, each $X_1$ is uniformly generated]
{ \includegraphics[width=0.23\textwidth]{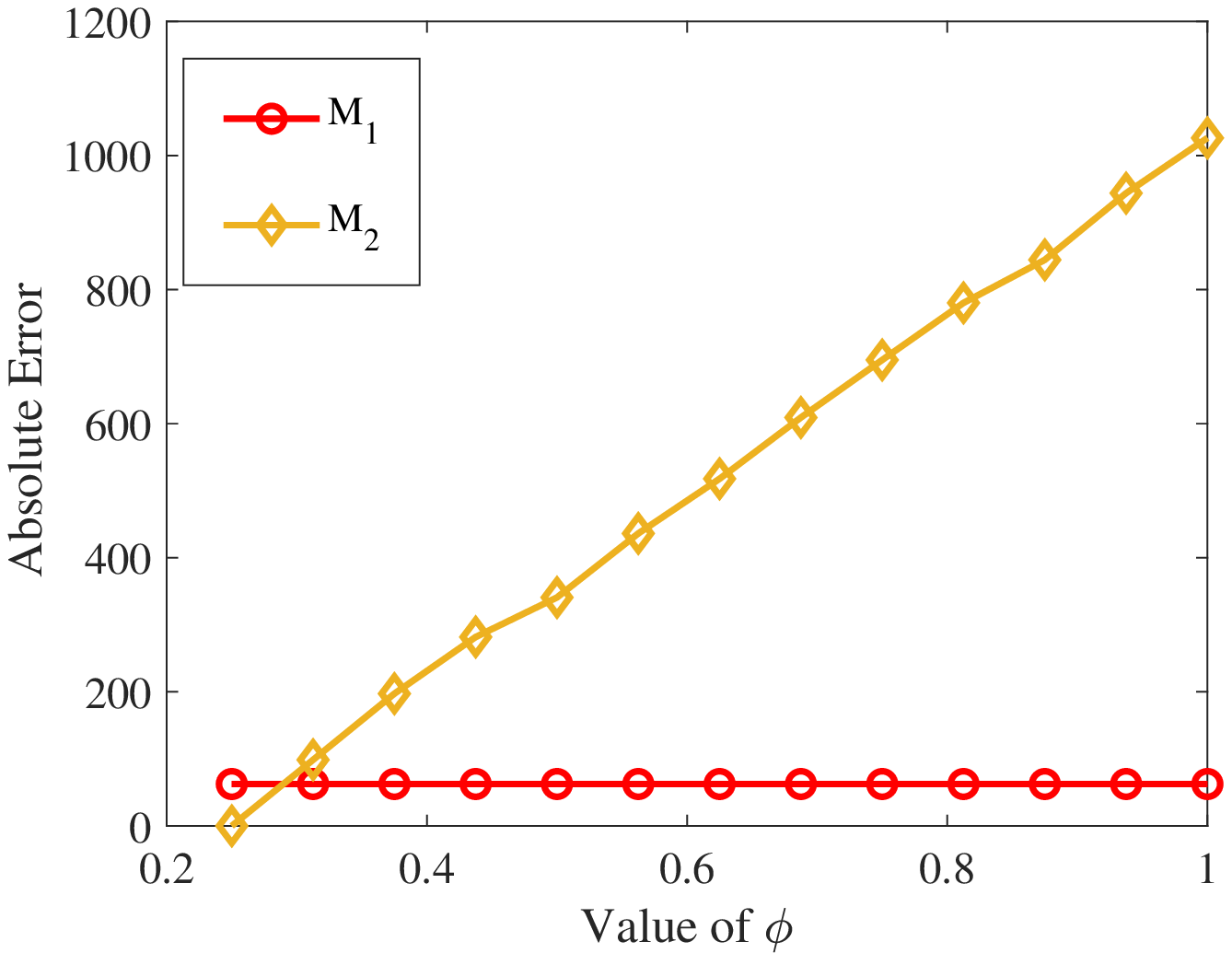}
\label{case 2} } 
\label{numerical1} 
\subfigure[$\mathcal{S}\cap\mathcal{L}=\emptyset$, $P_{X_1}(A)=0.8$] 
{ \includegraphics[width=0.23\textwidth]{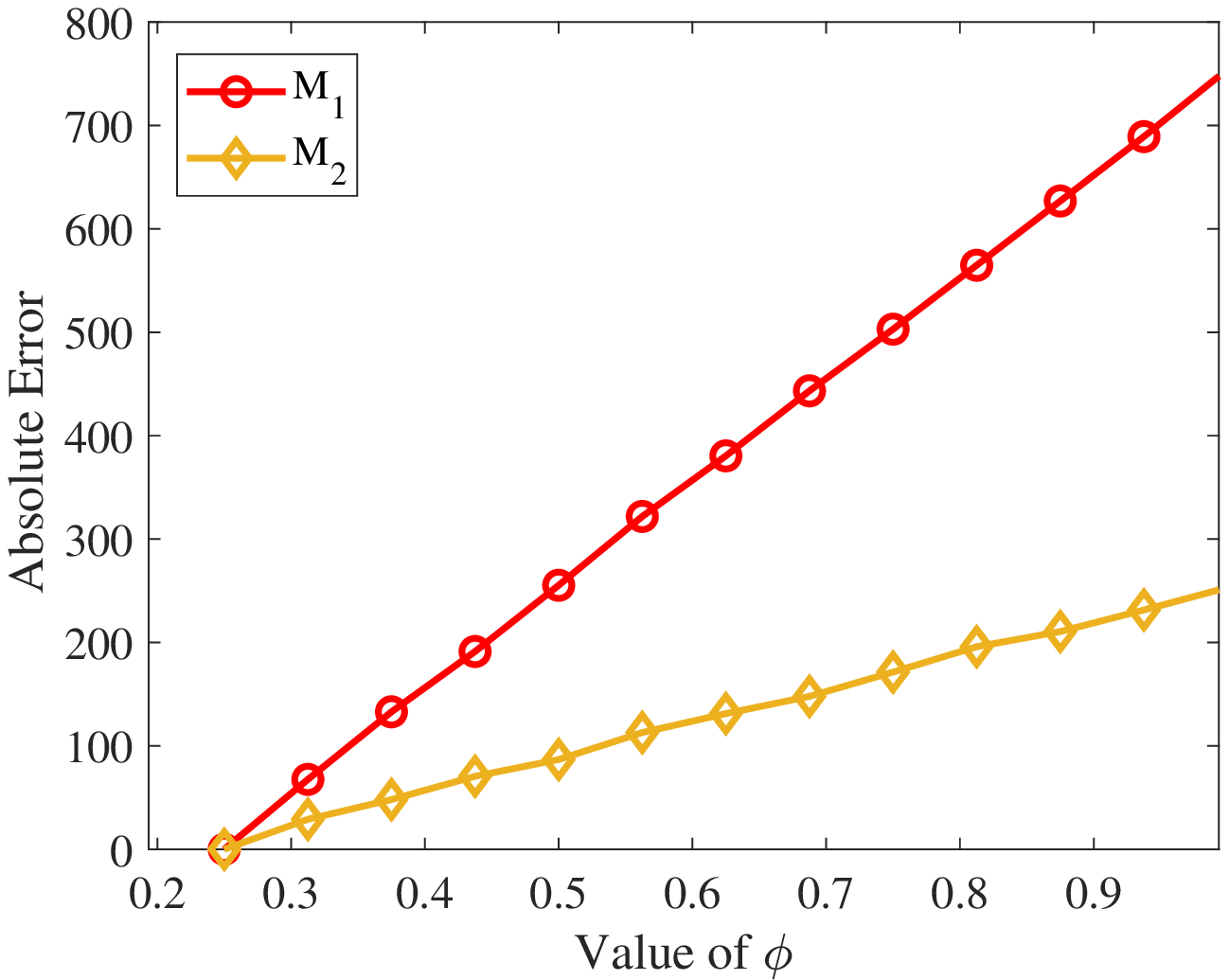} 
\label{case 3} }
\subfigure[$\mathcal{S}\cap\mathcal{L}\neq\emptyset$, $P_{X_1}(A)=0.8$]
{ \includegraphics[width=0.23\textwidth]{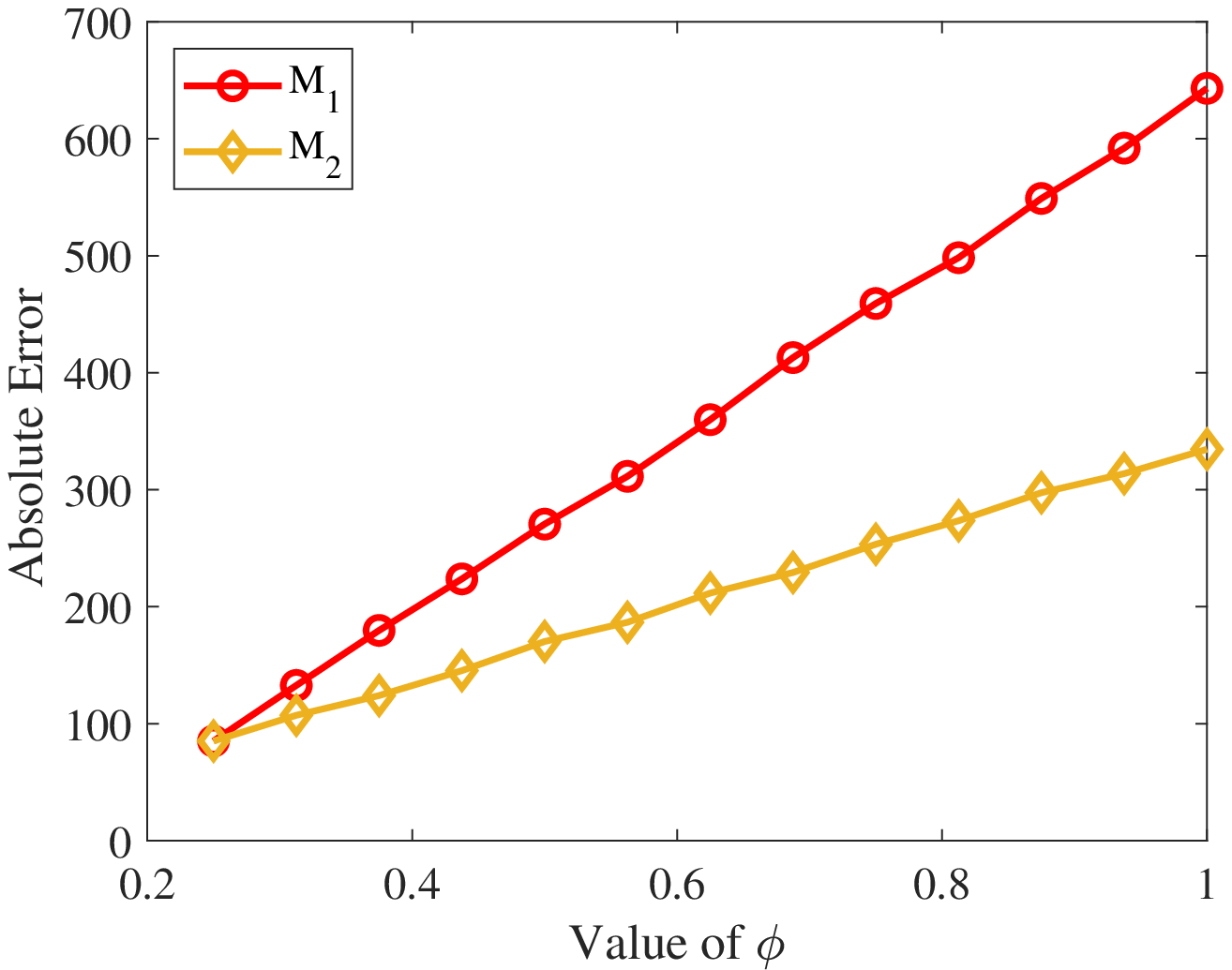}
\label{case 4} } 
\caption{Numerical results of the Markov model for different cases regarding the intersected length of $\mathcal{L}\cap\mathcal{S}$ and the distribution of $X_1$ for the local case.} 
\label{numerical2} 
\vspace{-10pt}
\end{figure*}

\subsection{Comparison of $LB(P_e)$ with $P_e$ from the Local Mechanism}

We next compare the $P_e$ resulting from the local mechanism with the lower bound in Theorem \ref{thm:lowerbound}. We consider three different cases: (i) when $\mathcal{S}\cap\mathcal{L}=\emptyset$, (ii) when $\mathcal{S}\cap\mathcal{L}=4$ (partially overlapped) and (iii) when $\mathcal{S}\cap\mathcal{L}=\mathcal{S}$ (fully overlapped). Under each case, we consider three types of correlation: low correlation $\phi=0.27$ (slightly skewed than uniform), moderate correlation $\phi=0.5$, and high correlation $\phi=0.95$ (slightly weaker than fully dependent). The comparisons are shown in Fig. \ref{lower}, where we plotted the lower bound of $P_e$ from Theorem \ref{thm:lowerbound}, $P_e$ calculated from Theorem \ref{thm:Pe} as well as the empirical $P_e$ averaged from $N$ samples, and we vary $N$ from $10$ to $1000$ to illustrate the convergence. Observe that the gap between $P_e$ from Theorem \ref{thm:Pe} and the lower bound of $P_e$ is different under different scenarios. Generally, extremely high or low correlation leads to smaller gaps. It is worth noting that the lower bound derived in this paper is based on Fano's inequality. There could be some $P_e$ that are not achievable, thus the gap between the $P_e$ from Theorem \ref{thm:Pe} and the optimal one (assume there is a mechanism that incurs a minimum $P_e$ and satisfies perfect privacy) can be even smaller.

\subsection{Comparison between $M_1$ and $M_2$.}
 Next, we examine the impact from (a) different prior distributions; (b) whether $\mathcal{S}\cap\mathcal{L}$ is an empty set or not. We fix $\mathcal{S}=\{3,4\}$ and evaluate with parameters according to the following cases: for (a), we assume $X_1$ is uniformly distributed or $P_{X_1}=0.8$, and then generate the whole sequence based on the value of $X_1$ and $\mathcal{T}$.  (b), we consider either ${\mathcal{L}}=\{4,5\}$ or ${\mathcal{L}}=\{5,6\}$. We compare the absolute error resulted from each mechanism under each case.

The results are shown in Fig. \ref{case 1}. Observe from Fig. \ref{case 1}, all mechanisms achieve zero-EAE when $\phi=0.25$. When $\phi\neq{0.25}$, $\mathcal{M}_1$ outperforms $\mathcal{M}_2$, the reason is that, to guarantee perfect privacy, $\mathcal{M}_1$ tends to release most of the local answers as $0$ while $\mathcal{M}_2$ releases most answers as $1$. Since $P_{X^{(k)}_{\mathcal{L}}}(v_{\mathcal{L}})$ happens with probability less than $1/2$, releasing more $0$s is more accurate than releasing more $1$s. Also, observe that the EAE of $\mathcal{M}_1$ first increases then decreases to $0$ with $\phi$. Intuitively, the value of $R(\bold{X},\mathcal{Q})$ decreases with $\phi$, and larger $R(\bold{X},\mathcal{Q})$ implies larger probability to directly release the answer. That is why EAE is smaller when $\phi$ is small. For large $\phi$, the dependence of data increases, and each user's genomic sequence can hardly match $v_{\mathcal{L}}$. As a result, the real local answer for each user is $0$. Thus, $\mathcal{M}_1$ which releases more $0$s achieves $0$-EAE when $\phi=1$. For case 2, from Fig. \ref{case 2}, the EAE for $\mathcal{M}_1$ is a constant, this is because when data is uniformly distributed, $\mathcal{E}=\frac{1}{4}$, which is smaller than $0.5$. As a result, $\mathcal{M}_1$ releases $0$ all the time, hence incurring a constant error for different $\phi$. 
It is worth noting that each case in Fig. \ref{case 2} provides lower $P_e$ than those in Fig. \ref{case 1}. The reason is that the prior of case 2 is more skewed and requires less perturbation.

Then, we examine how data prior affects the performance of $\mathcal{M}_1$ and $\mathcal{M}_2$. We let $P_{X_1}(A)=0.8$, and let $v_{\mathcal{L}}=\{A,A\}$, which means, each local answer $A^{(k)}$ is more likely to be $1$ than $0$, and as $\phi$ increases the probability of $P_{A^{(k)}}(1)$ increases. The comparison of the two mechanisms' performance is shown in Fig. \ref{case 3} and Fig. \ref{case 4}. We can observe that under case 1 and case 2, both mechanisms' EAEs increase as $\phi$ increases since larger $\phi$ implies smaller $R(\bold{X},\mathcal{Q})$. On the other hand, $\mathcal{M}_2$ performs better than $\mathcal{M}_1$, because each true local answer is more likely to be $1$, and $\mathcal{M}_2$ tends to release $1$ while $\mathcal{M}_1$ tends to release $0$ under high data correlations. It is worth noting that when $\phi=0.25$, the data is i.i.d. For case 1, each mechanism releases the true answer and achieves $0$-EAE; for case 2, $\mathcal{M}_1$ and $\mathcal{M}_2$ release $1$ when $X_{\bar{\mathcal{L}}}=v_{\bar{\mathcal{L}}}$, since $\mathcal{E}\le{0.5}$. As a result, an error is incurred when $X_{\mathcal{L}\cap\mathcal{S}}\neq{v_{\mathcal{L}\cap\mathcal{S}}}$ (each local $A^{(k)}=1$ when $X_{\bar{\mathcal{L}}}={v_{\bar{\mathcal{L}}}}$ and $X_{\mathcal{L}\cap\mathcal{S}}={v_{\mathcal{L}\cap\mathcal{S}}}$).

\subsection{Comparison between Centralized and Local Mechanism}
We next compare the EAEs incurred by the local and centralized mechanisms, consider $K$ independent users who hold genomic data, denote $\bold{X}^{(k)}$ as the $i$-th user's genomic data sequence with a length of $N$, i.e., $\bold{X}^{(k)}=\{X_1^{(k)},X_2^{(k)},...,X_N^{(k)}\}$. It is assumed that the prior of each $X_1^{(k)}$ is uniformly distributed, and the correlation of between the data from one to another follows Markov property, which can be summarized as follows: $P_{X_{k+1}^{(k)}|X_{k}^{(k)}}(\delta|\gamma)=b$, and $P_{X_{k+1}^{(k)}|X_{k}^{(k)}}(\gamma|\gamma)=a$, where $k$ denotes the index of different genomic data, $a\ge{b}$ and $\delta$, $\gamma$ are possible realizations from the support of the genomic data: $\{A,T,G,C\}$. Therefore: $a+3b=1$. Under the above setting, we first assume $\mathcal{S}=1,2$ and $\mathcal{L}=3,4$, $v=\{A,T\}$. That is, $A^{(k)}=1$ iff $X^{(k)}_{3,4}=\{A,T\}$.  We next show how different mechanisms perform under this setting.

\begin{figure}[t]
\centering 
\subfigure[] 
{ \includegraphics[width=0.22\textwidth]{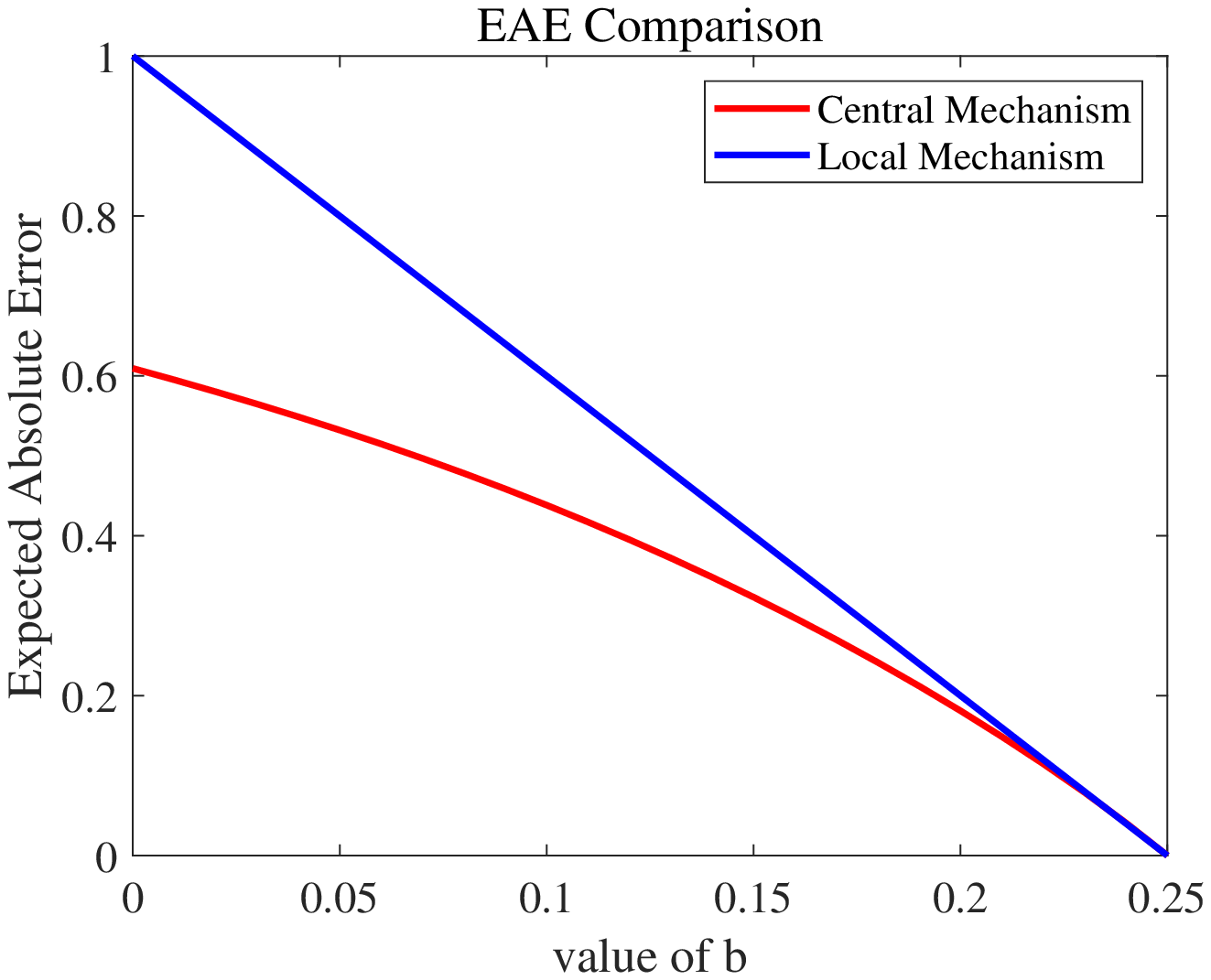}
\label{fig:my_label} }
\subfigure[] 
{ \includegraphics[width=0.22\textwidth]{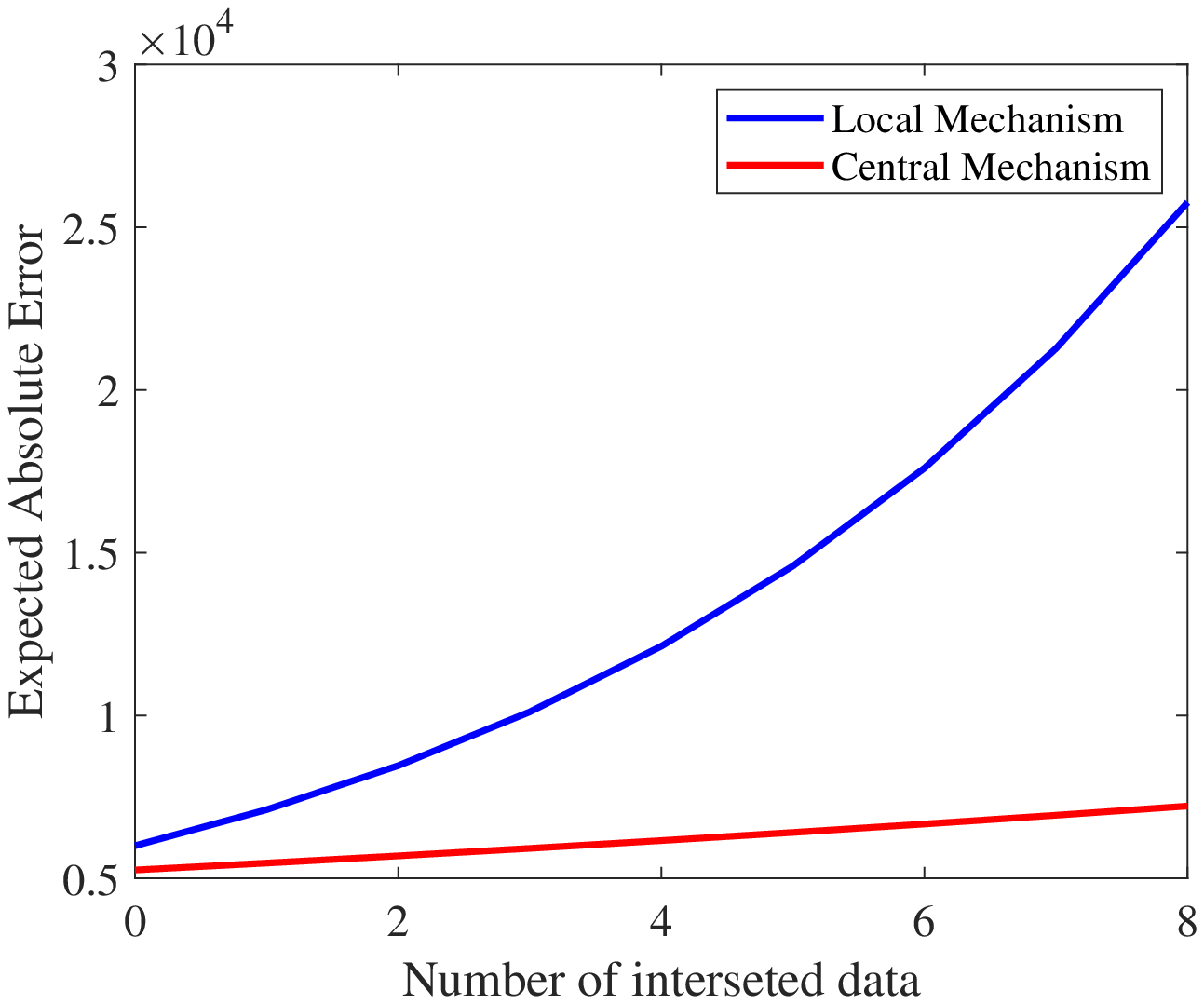}
\label{com2} }
\caption{ EAE comparison between the central and local mechanisms (a) with different correlation parameters (b) with different length of intersected data}
\label{inst} 
\vspace{-10pt}
\end{figure}

We vary the value of $b$ from $0$ (strong dependence) to $0.25$ (independent), and observe the EAEs resulted from different mechanisms. Observe that when $b=0.25$, which means data within the genomic sequence are independent, both mechanisms achieve zero expected absolute error; when $b=0$, referring to the strongest dependence scenario, the central mechanism always results in smaller EAE than the local mechanism.

We then fix the correlation to be $b=0.5$, and then vary the number of intersected data length from $0$ to $8$, i.e., $\mathcal{L}=\{3,4,5,6,7,8,9,10\}$, and increase the length of $\mathcal{S}$, from $\{1,2\}$ to $\{1,2,3,4,5,6,7,8,9,10\}$. We then compare the resulted expected absolute error, and the result is shown in Fig. \ref{inst}. Observe that, with the increase of the number of intersected data lengths, both mechanisms result in a larger expected error. However, the error of the local model increases much faster than the central model.

\subsection{Comparison with Differential Privacy-based Mechanisms}

In this part, we compare the performance of the proposed mechanism with (Local) Differential Privacy\cite{Extreme_ldp}. Since our mechanisms guarantee perfect privacy, equivalently, under the privacy notion of DP, $\epsilon=0$. Then, in the following, we show that under different $P_e$ or EAE, the minimum value of $\epsilon$ (a smaller $\epsilon$ indicates a stronger privacy guarantee) provided by different mechanisms. Under the local setting, we use a binary randomized response perturbation mechanism that satisfies LDP. In \cite{Extreme_ldp}, an optimal mechanism is derived under LDP constraints:  $\text{Pr}(Y=A)=\frac{e^{\epsilon}}{1+e^{\epsilon}}$, and $\text{Pr}(Y\neq A)=\frac{1}{e^{\epsilon}+1}$, wherein the binary case, the latter stands for the $P_e$. On the other hand, in \cite{Tianhao}, the minimum MSE resulted by the LDP-based mechanism for aggregated count query is given by:
$N\frac{e^{\epsilon}(e^{\epsilon}+1)}{(e^{\epsilon}-1)^2}$.
We next present the comparison under two cases: the comparison of the $P_e$ and the comparison of the MSE. The results are shown in Fig. \ref{com_DP}. Observe that under a larger $P_e$ or MSE, Differential Privacy can provide a better privacy guarantee with a small $\epsilon$. However, our proposed mechanisms always provide perfect privacy.

\begin{figure}[t]
\centering 
\subfigure[$\epsilon$ $v.s.$ $P_e$ (per-user case)] 
{ \includegraphics[width=0.22\textwidth]{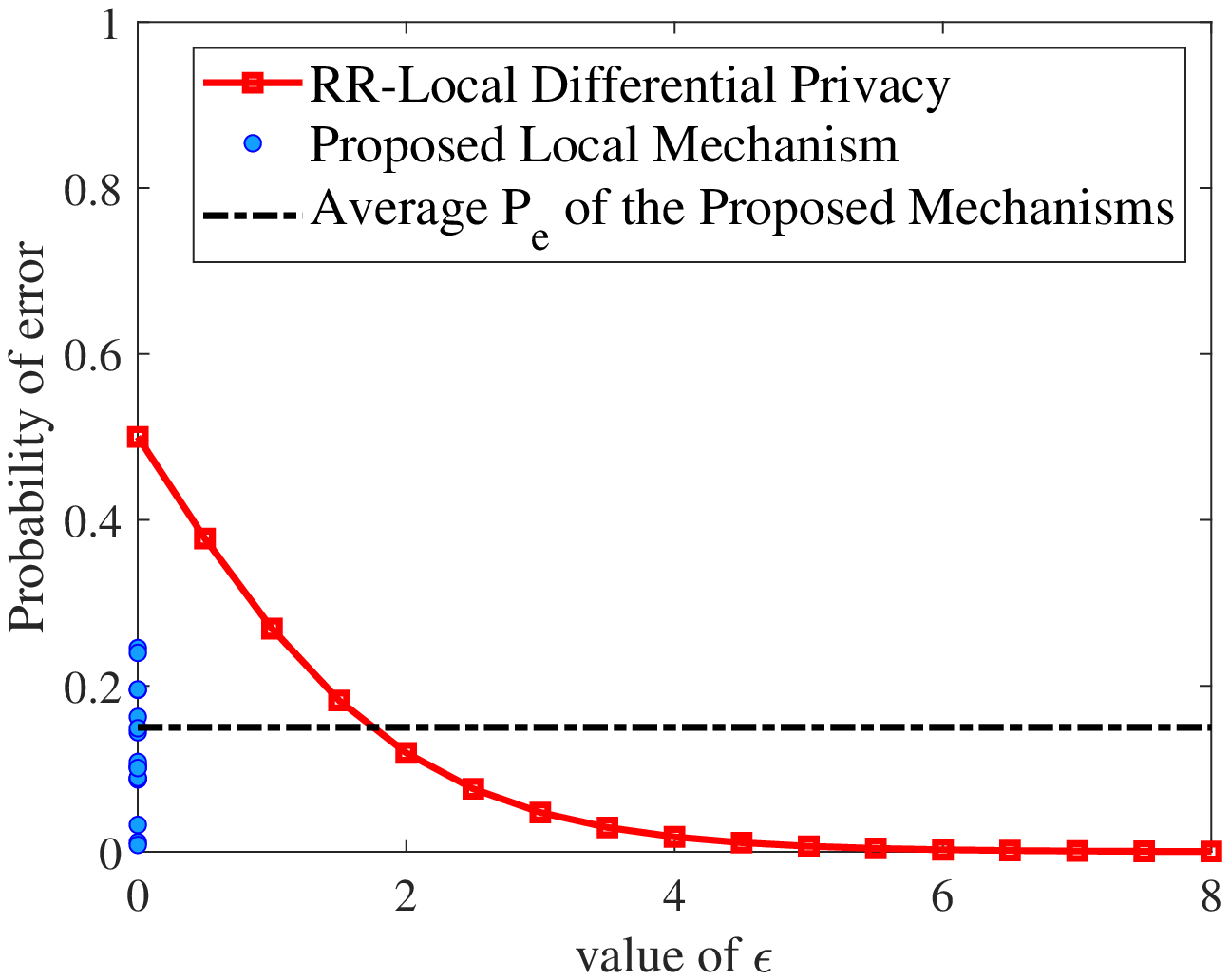} 
\label{LDP} }
\subfigure[$\epsilon$ $v.s.$ MSE (aggregated case)]
{ \includegraphics[width=0.22\textwidth]{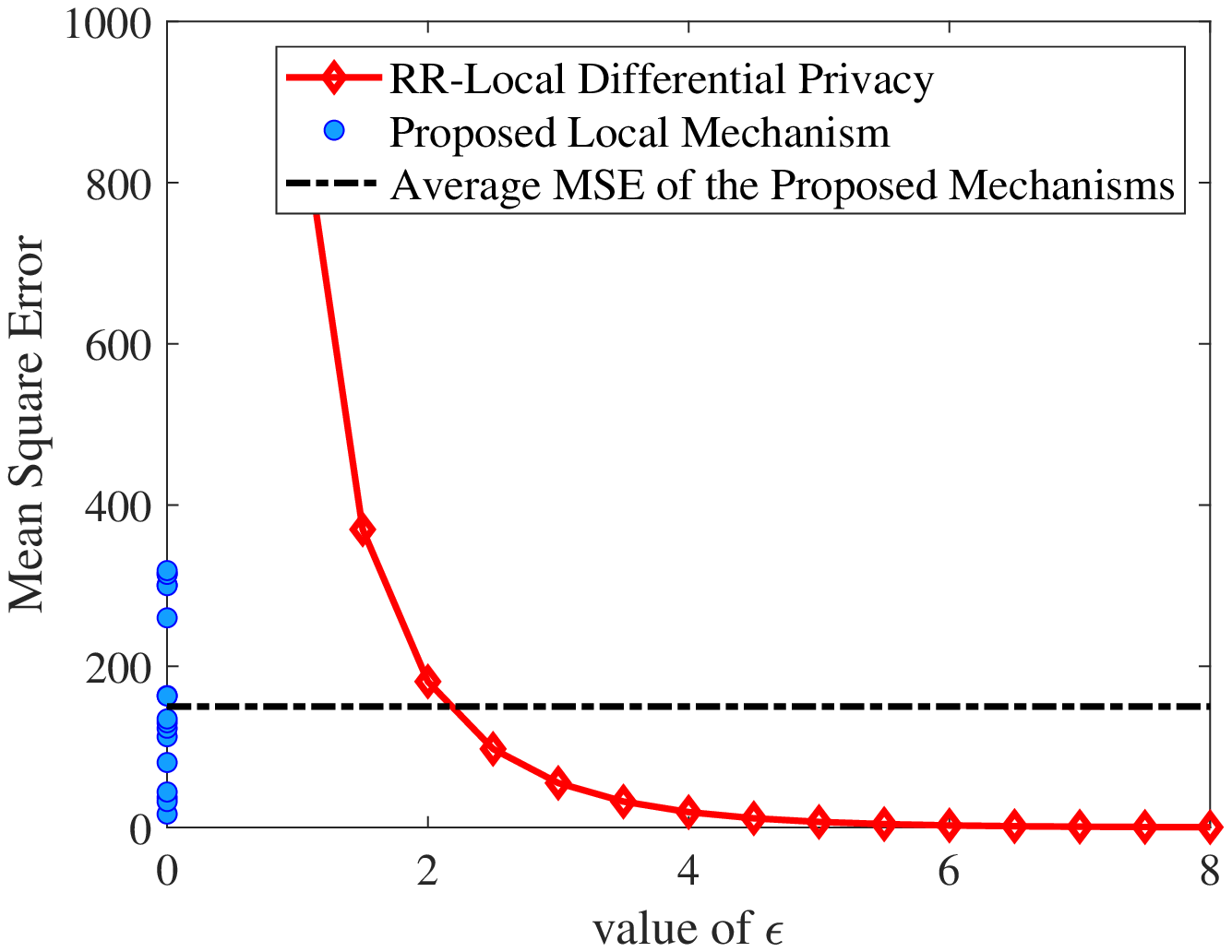}
\label{DP} } 
\caption{Comparison between the privacy guarantee by $\mathcal{M}_1$ or $\mathcal{M}_2$ and DP.} 
\label{com_DP} 
\vspace{-10pt}
\end{figure}

\begin{figure}[t]
    \centering
    {\includegraphics[width=1\columnwidth]{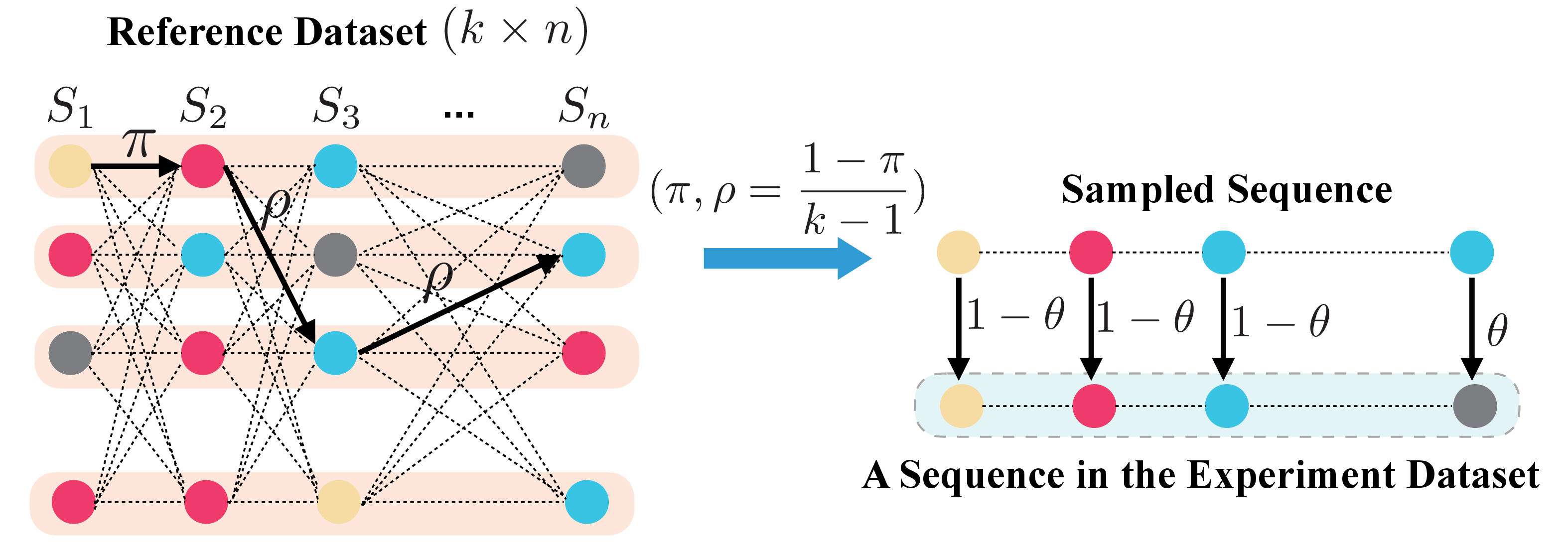}
    \caption{Generating experiment dataset from reference dataset with Hidden Markov Model (hmm)}
    \vspace{-10pt}
    \label{fig:hmm}}
\end{figure}

\begin{figure*}[htp]
\centering  
\subfigure[AAE comparison when $\theta = 0.01$]{ \includegraphics[width=0.23\textwidth]{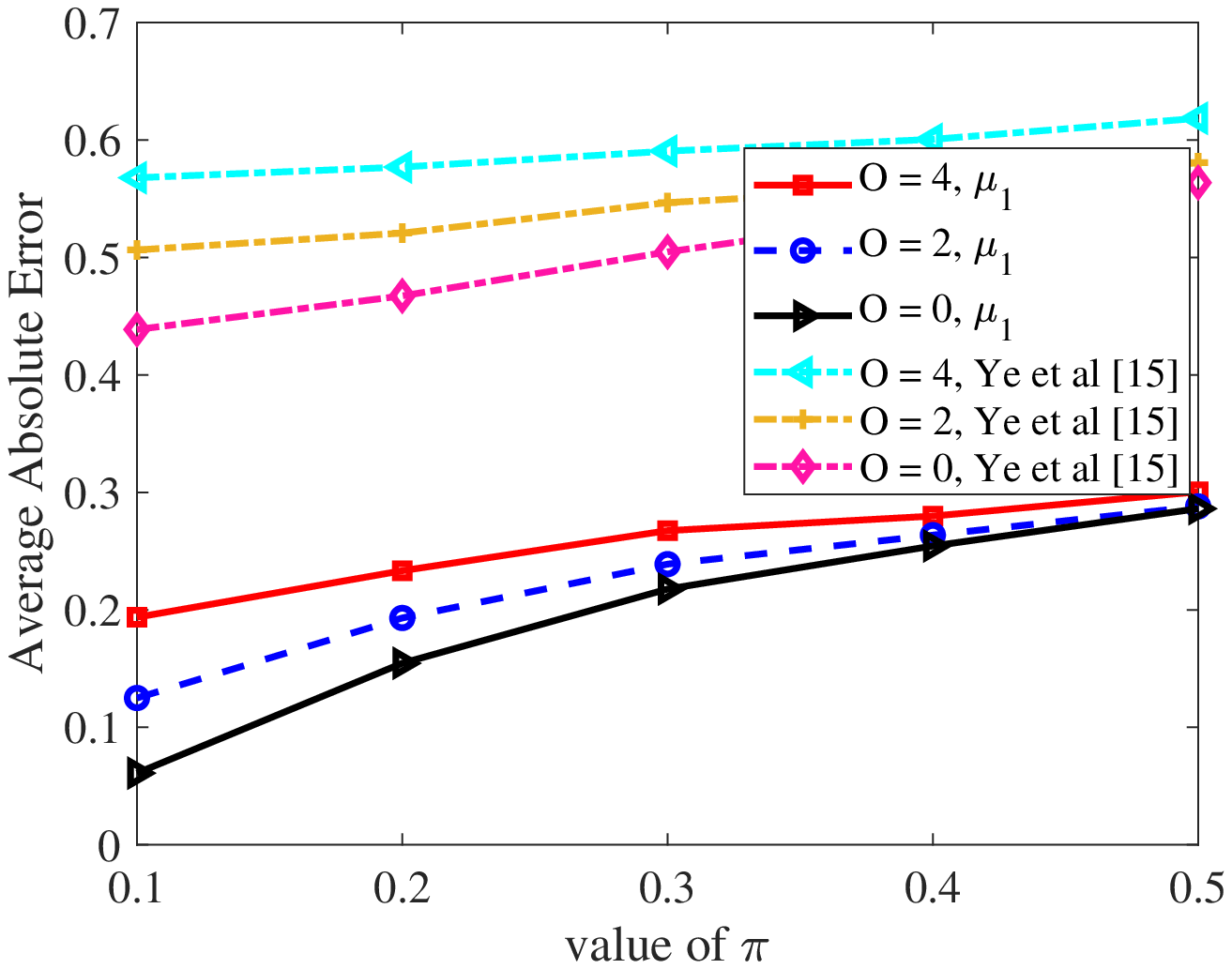} 
\label{hmm 1} }
\subfigure[AAE comparison when $\theta = 0.05$]
{ \includegraphics[width=0.23\textwidth]{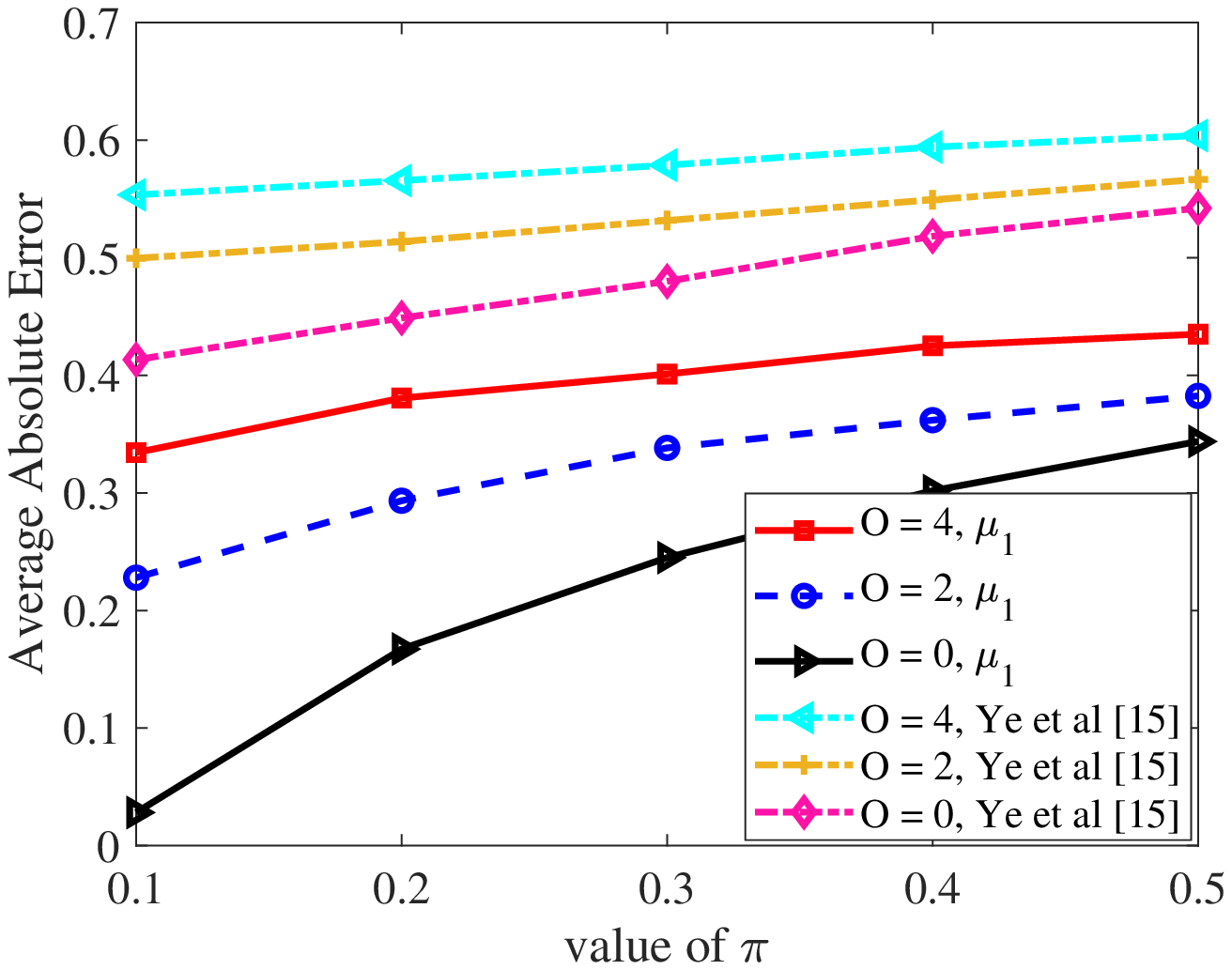}
\label{hmm 2} } 
\subfigure[EAE comparison when $\theta = 0.01$] 
{ \includegraphics[width=0.23\textwidth]{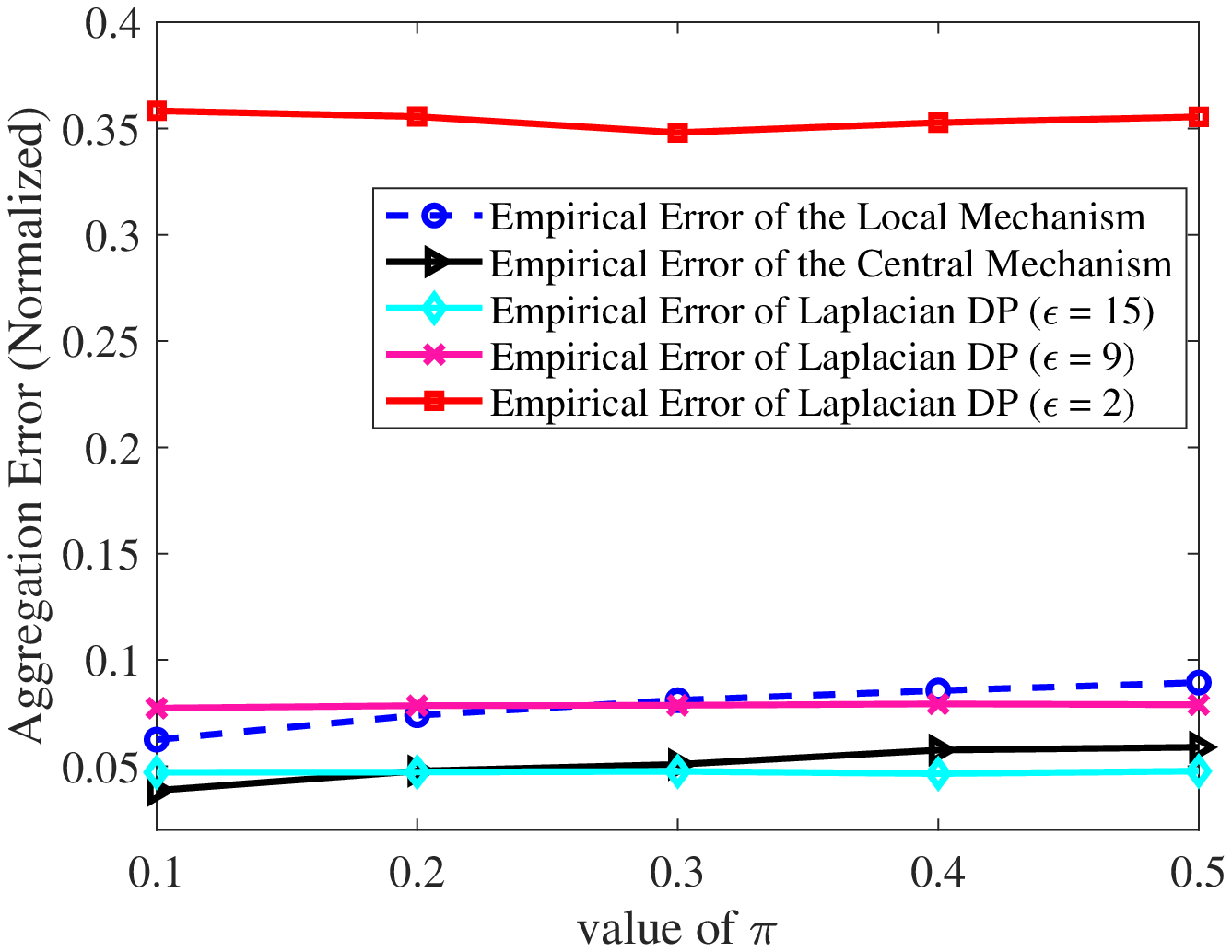} 
\label{hmm 3} }
\subfigure[EAE comparison when $\theta = 0.05$]
{ \includegraphics[width=0.23\textwidth]{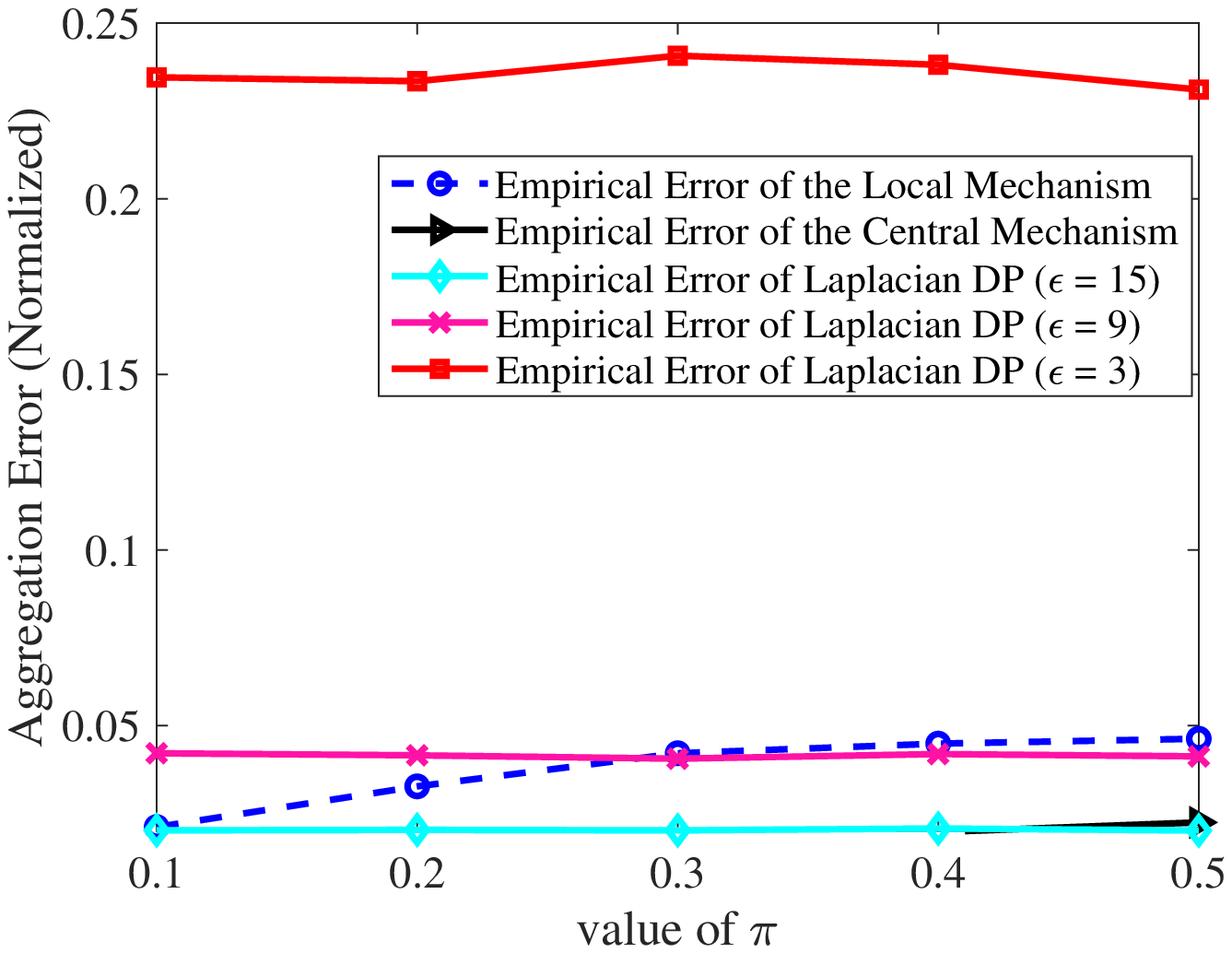}
\label{hmm 4} } 
\caption{Numerical results of the Hidden Markov Model generated from a synthetic reference dataset, (a), (b) for individual average absolute error (AAE) under different cases of the overlap $\mathcal{L}\cap\mathcal{S}$ (denoted as $O$), We sample from the prior for the genotypes at the hidden locations for Ye's mechanism in \cite{ye2020mechanisms}; (c), (d) for aggregate error, which compares empirical error in aggregation when deploying local and central mechanisms.} 
\label{numerical2} 
\vspace{-10pt}
\end{figure*}

\begin{figure*}[t]
\centering  
\subfigure[AAE comparison when $\theta = 0.01$]{ \includegraphics[width=0.23\textwidth]{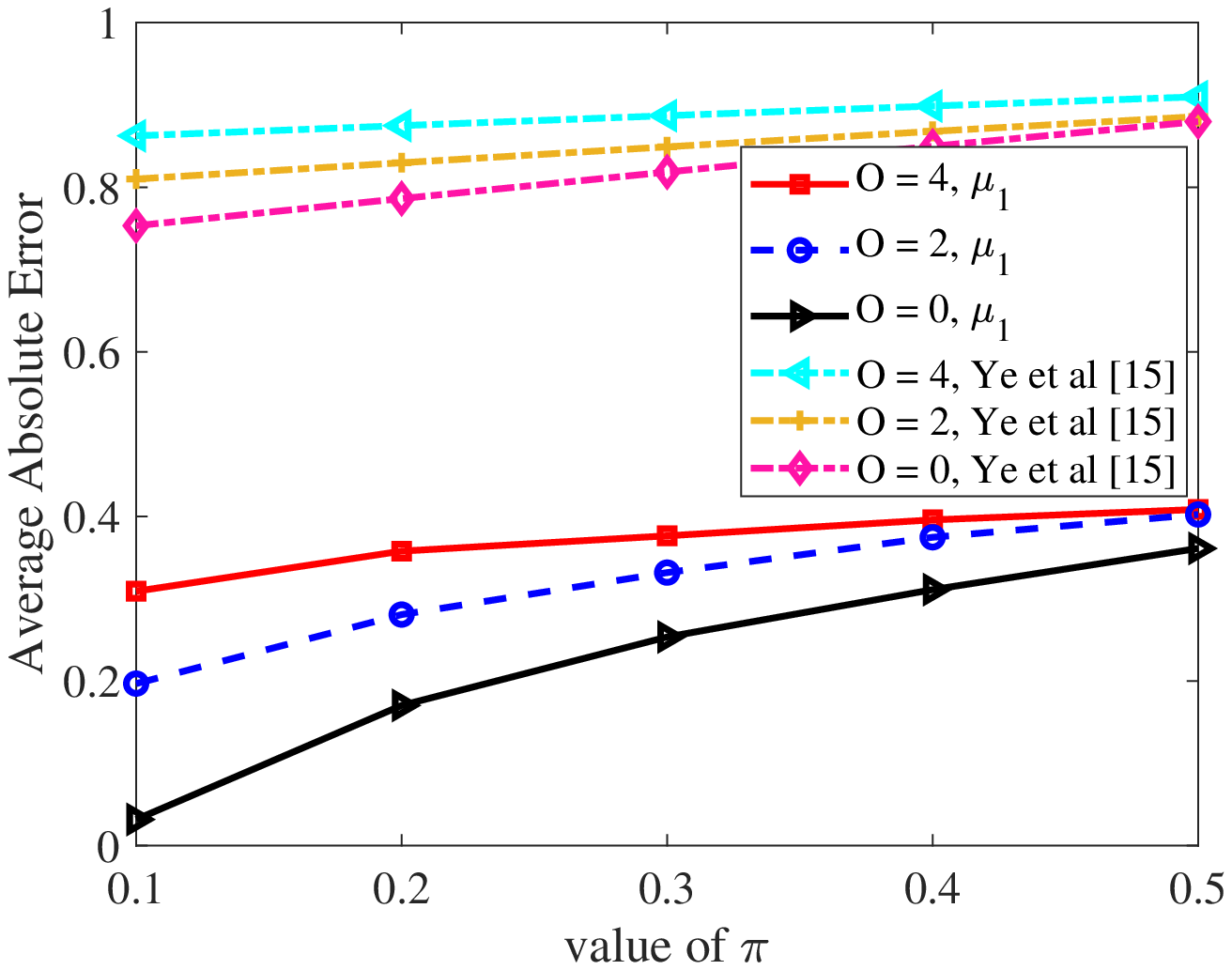} 
\label{real 1} }
\subfigure[AAE comparison when $\theta = 0.05$]
{ \includegraphics[width=0.23\textwidth]{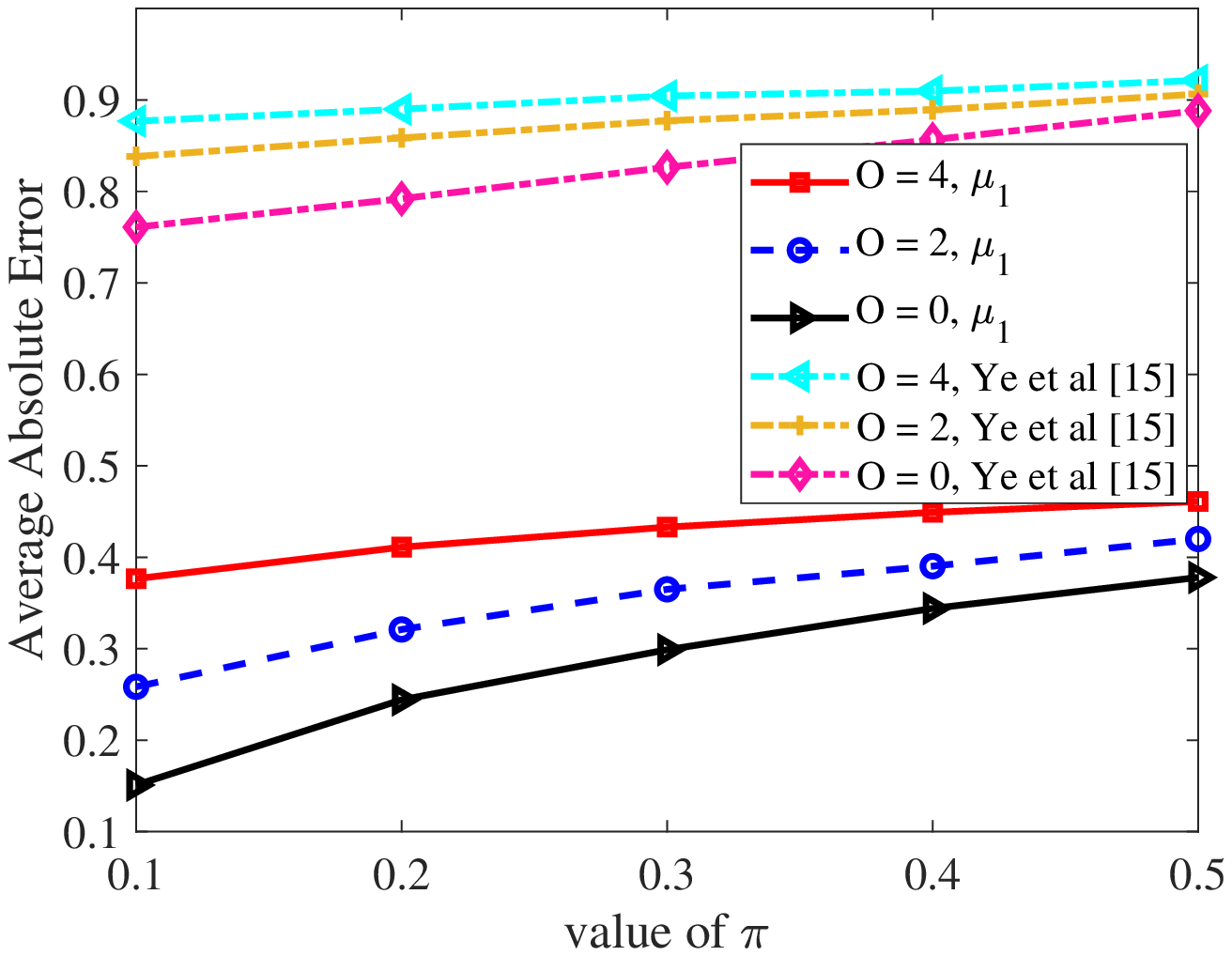}
\label{real 2} } 
\subfigure[EAE comparison when $\theta = 0.01$] 
{ \includegraphics[width=0.23\textwidth]{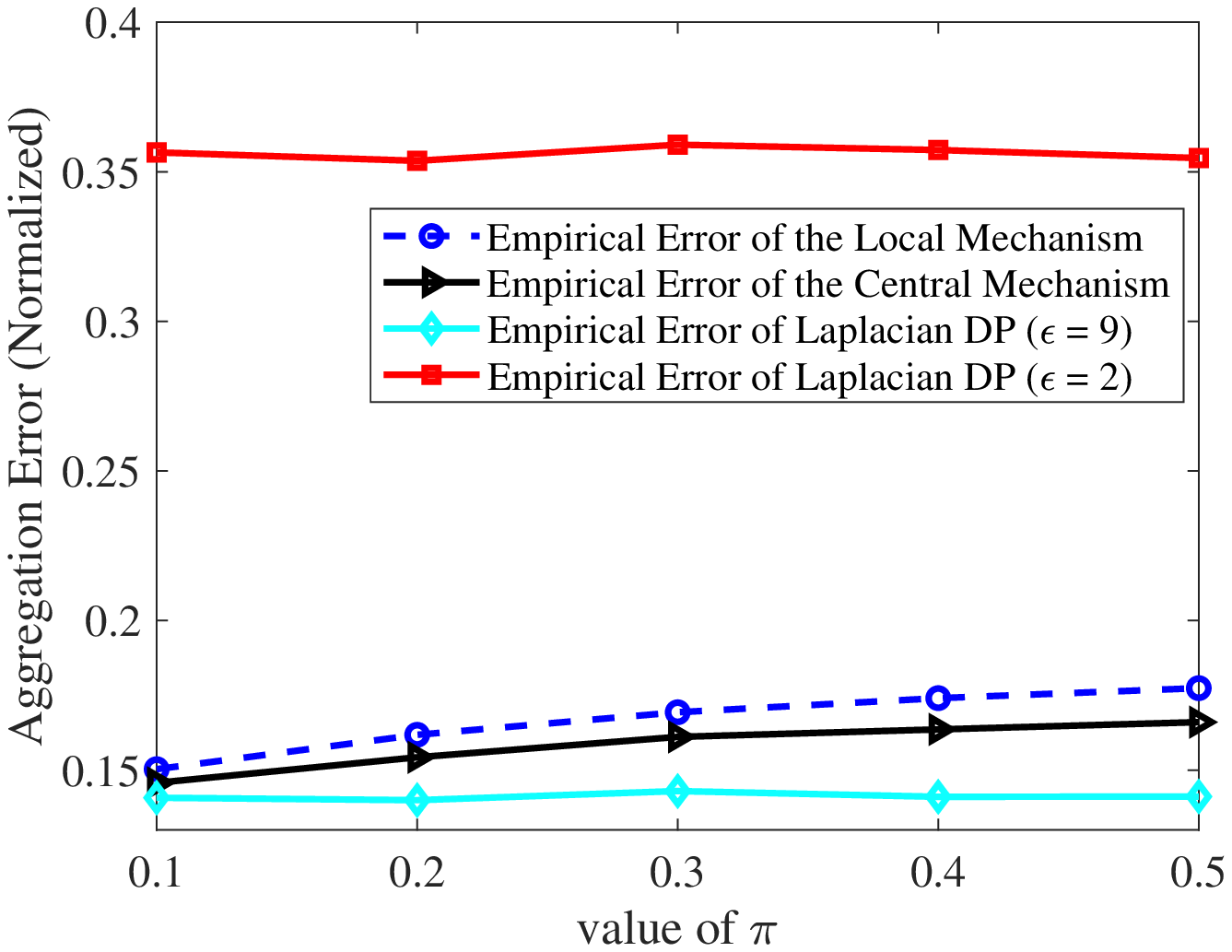} 
\label{real 3} }
\subfigure[EAE comparison when $\theta = 0.05$]
{ \includegraphics[width=0.23\textwidth]{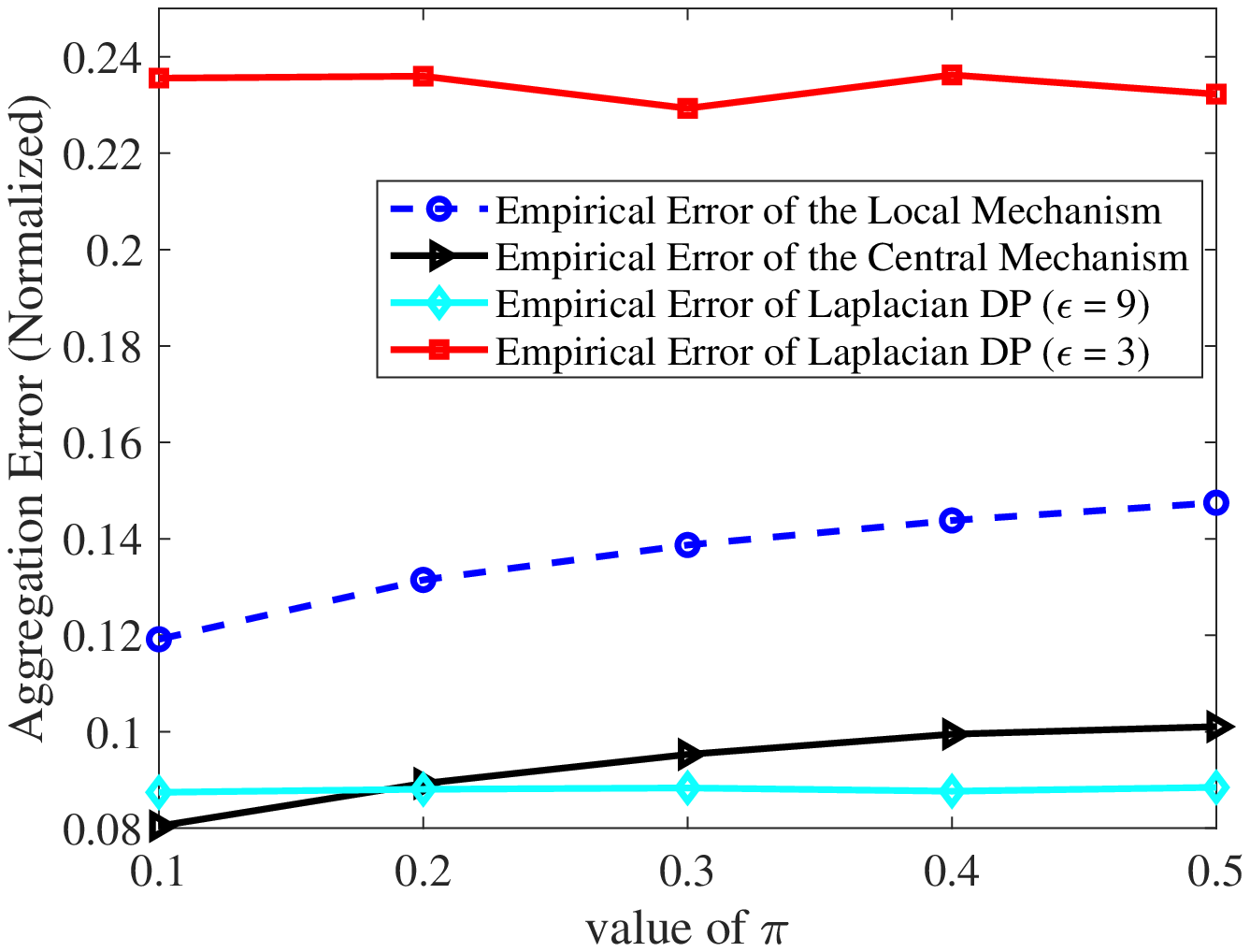}
\label{real 4} } 
\caption{Numerical results of the Hidden Markov Model generated from real-world genomic data, (a), (b) for individual average absolute error (AAE) under different cases of $\mathcal{L}\cap\mathcal{S}$(denoted as $O$), We sample from the uniform distribution for the genotypes at the hidden locations for Ye's mechanism in \cite{ye2020mechanisms}; (c), (d) for aggregate error, which compares empirical error in aggregation when deploying local and central mechanisms.} 
\label{numerical2} 
\vspace{-10pt}
\end{figure*}

\subsection{Evaluation with Hidden Markov Model}

Since the real-world genomic dataset contains private information and is infeasible to access. We next evaluate with the Hidden Markov Model (HMM), which is widely adopted in genetics for a wide range of tasks that
require a probabilistic model of the genome \cite{ye2020mechanisms, Ligenetics, Song_yun}. In the HMM model, the experimental data (set) is generated from a reference dataset with certain parameters $(\pi, \theta)$. In this experiment, we consider the reference dataset contains  { $K = 100$ users genomic sequence, each with a length of $N = 20$}. We treat each genotype at a certain location as a state. When generating the experiment data sequence, we randomly pick one data sequence from the reference dataset to begin with. For each of the following genotypes, we make it identical to the state of the same sequence with a probability of $\pi$, and switch to another state from another sequence with a probability of $(1-\pi)/(N-1)$. It is worth noting that the switching probability of $\pi$ can be used to measure the dependence in the genomic sequence. A larger value of $\pi$ indicates that genomic sequences in the experimental dataset would preserve similar patterns to the reference dataset. Thus the dependence is strong. On the other hand, for a small $\pi$, each genotype in the experiment dataset contains large randomness, and hence the dependence is weak. Besides the switching probability, we add randomness to the experimental dataset with an error probability $\theta$. That said, each time HMM model copies a genotype from the reference dataset, there is a probability of $\theta$ to substitute by randomly sampling one from $\{A,T,G,C\}$. Finally, we compare it with the mechanism proposed in \cite{ye2020mechanisms}, which is the most relative mechanism in the literature. Note that the mechanism in \cite{ye2020mechanisms} hides genotypes at certain locations, which makes the query aggregation not directly applicable. To this end, we consider two types of post-processing strategies: one is to sample genotypes at hidden locations from a uniform distribution, and the other is to sample from the prior distribution calculated from the frequency. Figure \ref{fig:hmm} depicts the process that an HMM model generates an experimental dataset from the reference dataset.

To generate the experimental dataset, we consider two sets of reference dataset, the first reference dataset is sampled from a uniform distribution. The second reference dataset is a sample from Sample GenBank Record \cite{sample}. We generate $K = 1000$ users' genomic sequence as the experimental data with $\pi$ and $\theta$. In the following experiments, we vary the value of $\pi$ from $0$ to $0.5$, and fix the value of $\theta$ to be $0.01$ and $0.05$, respectively. Then, we calculate the frequency of the appearance of different combinations of genotypes as the prior distribution and the conditional probabilities. Specifically, to calculate the distribution of $A$ in the central setting, we find the ratio of $\bar{A}/K$, where $\bar{A}$ denotes the true aggregate value, this value converges to $P_{A^{(k)}}(1)$ for each individual (for large $K$). Under the independent user assumption (each user in the experimental dataset is generated by HMM independently), $A$ is binomially distributed. Hence, its prior distribution can be calculated accordingly by each $P_{A^{(k)}}$. From there, we implement our local and central mechanisms accordingly for $T=1000$ times and calculate the average values.

We first compare the local models with different length of $|\mathcal{L}\cap{\mathcal{S}}|$, the results are shown in Fig.\ref{hmm 1} and Fig.\ref{hmm 2} respectively (Fig.\ref{real 1} and Fig.\ref{real 2} for real-world sample). Observer that, as $\pi$ increases, the $P_e$ of each case increases accordingly. That is because the dependence increases with $\pi$, and $P_e$ increases with the strength of the dependence. Another observation is that $P_e$ also increases with the length of intersected locations. The reason is the dependence between the queried genotypes and the sensitive genotypes increases with the length of intersected locations. It is worth noting that for larger $\theta$, the randomness in the dataset increases, which results in smaller dependence among the genomic sequence, and hence incurs a smaller $P_e$, which explains why sub-case $(b)$ always incurs smaller errors than subcase $(a)$.

We then compare the local model and the central model according to different $\theta$ and $\pi$. In addition, we compared our models with those based on DP with Laplace mechanisms. For each model, we modify the value of the privacy budget $\epsilon$, and compute the empirical absolute error by implementing $\epsilon$-DP. Then we find those mechanisms that incur similar aggregation errors to our mechanisms. The corresponding values of $\epsilon$ and the comparison results are shown in Fig. \ref{hmm 3} and Fig. \ref{hmm 4}, respectively (Fig.\ref{real 3} and Fig.\ref{real 4} for real-world sample). Observe that the central model always provides smaller empirical errors compared to the local models. On the other hand, to achieve comparable performance to our mechanism, the DP-based mechanism causes too much privacy leakage ($\epsilon$ is greater than $4$).  {It is worth noting that, the $\epsilon$ values yielding comparable accuracy to our local/central mechanisms were found to be $15$ and $9$ for synthetic and real-world data, respectively. That is to say, DP mechanisms need different amounts of privacy budget to achieve comparable utility (in terms of error probability) to our mechanisms for  synthetic and real-world data respectively. The primary reason for this difference is rooted in the inherent characteristics of real-world data sequences, where each local $A^{(k)}$ is more likely to be $0$. In view of privacy preservation, our mechanisms may introduce some noise by perturbing certain $0$s to $1$s. Consequently, this perturbation leads to an increase in the accumulated error. Another observation is when choosing DP with a reasonable budget for privacy, such as $2$ or $3$, the  accuracy gaps with our mechanisms are significant.}

\section{Conclusion}

 The problem of privacy-preserving counting query-answering mechanisms for genotype aggregation is studied in this paper. We propose information theoretical privacy-preserving mechanisms based on both the local and the central settings. The proposed mechanisms guarantee perfect privacy for genome data at sensitive locations. We then derive a lower bound of $P_e$ for an arbitrary mechanism under perfect privacy and show the optimality of our mechanisms under some special cases. Finally, we simulate with both synthetic data and real-world data sample combination with Hidden Markov Model to show the performance of our proposed mechanisms under different scenarios. The results also contain a comparison with Differential Privacy-based mechanisms, which shows that to provide comparable query accuracy, DP-based mechanisms incur much larger privacy leakages.

\appendices

\section{Proof of Theorem \ref{thm:Pe} } \label{appendix:proof_theorem2}

In this Section, we provide some more intuition behind the derivation of the mechanisms and then derive the error probability for the mechanisms. Since the proof for Theorem 2 is for local mechanism only, we remove the superscript $(k)$ for simplicity.

\subsection{Deriving the mechanisms and Error Probability}\label{appendix:derivations_mechanisms}

We first expand the term $P_{Y|\bold{X}_{\mathcal{S}}}(y|x_{\mathcal{S}})$ as follows
\begin{align}
    P_{Y|\bold{X}_{\mathcal{S}}}(y | x_{\mathcal{S}}) = \sum_{x_{\bar{\mathcal{L}}}} P_{Y|\bold{X}_{\bar{\mathcal{L}}}, \bold{X}_{\mathcal{S}}}(y | x_{\bar{\mathcal{L}}}, x_{\mathcal{S}}) P_{\bold{X}_{\bar{\mathcal{L}}}|\bold{X}_{\mathcal{S}}}(x_{\bar{\mathcal{L}}}|x_{\mathcal{S}}). \label{eqn:expanedTerm_xl}
\end{align}
In order to satisfy the per-user perfect privacy condition, we require that $P_{Y|\bold{X}_{\bar{\mathcal{L}}},\bold{X}_{\mathcal{S}}}(y | x_{\bar{\mathcal{L}}}, x_{\mathcal{S}}) \text{Pr}(x_{\bar{\mathcal{L}}}|x_{\mathcal{S}}) = f(y, x_{\bar{\mathcal{L}}}), \forall x_{\bar{\mathcal{L}}}, y$, i.e., each term in \eqref{eqn:expanedTerm_xl} does not depend on $x_{\mathcal{S}}$. To this end, we have
\begin{align}
    P_{Y|\bold{X}_{\bar{\mathcal{L}}}, \bold{X}_{\mathcal{S}}}(y | x_{\bar{\mathcal{L}}}, x_{\mathcal{S}}) = \frac{f(y, x_{\bar{\mathcal{L}}})}{P_{\bold{X}_{\bar{\mathcal{L}}}|\bold{X}_{\mathcal{S}}}(x_{\bar{\mathcal{L}}} | x_{\mathcal{S}})}, y \in \{0, 1\}. \label{eqn:functional_form}
\end{align}
From the above equation, we make $  P_{Y|\bold{X}_{\bar{\mathcal{L}}}, \bold{X}_{\mathcal{S}}}(y | x_{\bar{\mathcal{L}}}, x_{\mathcal{S}}) $ a valid probability mass function, by picking $f(y, x_{\bar{\mathcal{L}}} )$ as follows
\begin{align}
    &0 \leq f(y, x_{\bar{\mathcal{L}}} )\leq P_{\bold{X}_{\bar{\mathcal{L}}}|X_{\mathcal{S}}}(x_{\bar{\mathcal{L}}} | x_{\mathcal{S}}), \forall x_{\bar{\mathcal{L}}}, x_{\mathcal{S}}, \nonumber \\
    &\Rightarrow f(y, x_{\bar{\mathcal{L}}}) \leq \min_{w} P_{\bold{X}_{\bar{\mathcal{L}}}|\bold{X}_{\mathcal{S}}}(x_{\bar{\mathcal{L}}} | x_{\mathcal{S}} = w). \label{eqn:upperbound_mapping}
\end{align}

\subsection{Error Probability for Mechanism $\mathcal{M}_{1}$:}
Recall that our goal is to minimize the probability of error per user. We first expand the first term  as follows:
\begin{align}
    & P_{Y, \bold{X}_{\mathcal{L}}}(0, v_{\mathcal{L}} )  = \sum_{x_{\bar{\mathcal{S}}}} P_{Y| \bold{X}_{\mathcal{L}}, \bold{X}_{\bar{\mathcal{S}}}}(0| v_{\mathcal{L}}, x_{\bar{\mathcal{S}}} ) P_{X_{\mathcal{L}}, X_{\bar{\mathcal{S}}}}(v_{\mathcal{L}}, x_{\bar{\mathcal{S}}}) \nonumber \\
    &  = \sum_{x_{\bar{\mathcal{S}}}} P_{Y| X_{\mathcal{L}}, X_{\bar{\mathcal{S}}}}(0|  v_{\mathcal{L}}, x_{\bar{\mathcal{S}}} ) P_{X_{\mathcal{L}},X_{\bar{\mathcal{S}}}}(v_{\mathcal{L}}, x_{\bar{\mathcal{S}}}) \nonumber \\
    & =  \sum_{x_{\bar{\mathcal{S}}}} (1-P_{Y| \bold{X}_{\mathcal{L}}, \bold{X}_{\bar{\mathcal{S}}}}(1|  x_{\mathcal{L}}, x_{\bar{\mathcal{S}}} ) ) P_{\bold{X}_{\mathcal{L}},\bold{X}_{\bar{\mathcal{S}}}}(v_{\mathcal{L}}, x_{\bar{\mathcal{S}}}) \nonumber \\
      & =  P_{\bold{X}_{\mathcal{L}}}(v_{\mathcal{L}}) - \sum_{x_{\bar{\mathcal{S}}}} P_{Y| \bold{X}_{\mathcal{L}}, \bold{X}_{\bar{\mathcal{S}}}}(1|  v_{\mathcal{L}}, x_{\bar{\mathcal{S}}} )  P_{\bold{X}_{\mathcal{L}},\bold{X}_{\bar{\mathcal{S}}}}(v_{\mathcal{L}}, x_{\bar{\mathcal{S}}}) \nonumber \\
       & \overset{(a)} =  P_{\bold{X}_{\mathcal{L}}}(v_{\mathcal{L}}) - \sum_{x_{\bar{\mathcal{S}}}} P_{Y| \bold{X}_{\bar{\mathcal{L}}}, \bold{X}_{\mathcal\mathcal\mathcal\mathcal\mathcal\mathcal\mathcal{S}}}(1|  v_{\bar{\mathcal{L}}}, x_{\mathcal\mathcal\mathcal\mathcal\mathcal\mathcal\mathcal{S}} )  P_{\bold{X}_{\bar{\mathcal{L}}},\bold{X}_{\mathcal\mathcal\mathcal\mathcal\mathcal\mathcal\mathcal{S}}}(v_{\bar{\mathcal{L}}}, x_{\mathcal\mathcal\mathcal\mathcal\mathcal\mathcal\mathcal{S}}) \nonumber \\
         & \overset{(b)} =  P_{\bold{X}_{\mathcal{L}}}(v_{\mathcal{L}}) - \sum_{x_{\bar{\mathcal{S}}}} f(1, v_{\bar{\mathcal{L}}}) P_{\bold{X}_{\mathcal{S}}}(x_{\mathcal{S}}). \nonumber \\
             %& \overset{(c)}=  P_{X_{\mathcal{L}}}(v_{\mathcal{L}}) - f(1, v_{\bar{\mathcal{L}}}) \sum_{x_{\bar{\mathcal{S}}}} P_{X_{\mathcal{S}}}(x_{\mathcal{S}}) \nonumber \\
     %& = P_{X_{\mathcal{L}}}(v_{\mathcal{L}}) -   f(1, v_{\bar{\mathcal{L}}}) \sum_{x_{\bar{\mathcal{S}}}}  P_{X_{\bar{\mathcal{S}}}, X_{\mathcal{L} \cap \mathcal{S}}}(x_{\bar{\mathcal{S}}}, v_{\mathcal{L} \cap \mathcal{S}} ) \nonumber \\ 
      & = P_{\bold{X}_{\mathcal{L}}}(v_{\mathcal{L}}) -   f(1, v_{\bar{\mathcal{L}}})  P_{\bold{X}_{\mathcal{L} \cap \mathcal{S}}}(v_{\mathcal{L} \cap \mathcal{S}} ) 
    \label{eqn:error_term1}
\end{align}
where step (a) follows from the fact that $\mathcal{L} \cap \bar{\mathcal{S}} = \bar{\mathcal{L}} \cap \mathcal{S}$, while step (b) follows from \eqref{eqn:functional_form}. %In step (c), the function $f(0, v_{\bar{\mathcal{L}}}) $ does not depend on $x_{{\mathcal{S}}}$. 

Similarly, we have the following set of steps for the second term  as follows:
\begin{align}
    &\sum_{x_{\mathcal{L}} \neq v_{\mathcal{L}}}  P_{Y, X_{\mathcal{L}}}(1,  x_{\mathcal{L}} ) \nonumber\\
    & = \sum_{x_{\mathcal{L}} \neq  v_{\mathcal{L}}} \sum_{x_{\bar{\mathcal{S}}}} P_{Y| X_{\mathcal{L}}, X_{\bar{\mathcal{S}}}}(1|  x_{\mathcal{L}}, x_{\bar{\mathcal{S}}} ) P_{X_{\mathcal{L}},X_{\bar{\mathcal{S}}}}(x_{\mathcal{L}}, x_{\bar{\mathcal{S}}})  \nonumber \\  
    &  = \sum_{x_{\mathcal{L}} \neq v_{\mathcal{L}}} \sum_{x_{\bar{\mathcal{S}}}} P_{Y| X_{\bar{\mathcal{L}}}, X_{\mathcal{S}}}(1|  x_{\bar{\mathcal{L}}}, x_{\mathcal{S}} ) P_{X_{\bar{\mathcal{L}}}, X_{\mathcal\mathcal\mathcal\mathcal\mathcal\mathcal\mathcal{S}}}(x_{\bar{\mathcal{L}}}, x_{\mathcal{S}}) \nonumber \\
        %& \overset{(a)} = \sum_{x_{\mathcal{L}} \neq  v_{\mathcal{L}}} \sum_{x_{\bar{\mathcal{S}}}} f(1, x_{\bar{\mathcal{L}}}) P_{X_{\mathcal\mathcal\mathcal\mathcal\mathcal\mathcal\mathcal{S}}}(x_{\mathcal\mathcal\mathcal\mathcal\mathcal\mathcal\mathcal{S}}) \nonumber \\
        & \overset{(b)} = \sum_{x_{\mathcal{L}} \neq  v_{\mathcal{L}}} f(1, x_{\bar{\mathcal{L}}}) \sum_{x_{\bar{\mathcal{S}}}}  P_{\bold{X}_{\mathcal\mathcal\mathcal\mathcal\mathcal\mathcal\mathcal{S}}}(x_{\mathcal\mathcal\mathcal\mathcal\mathcal\mathcal\mathcal{S}}), \nonumber \\ 
        & =   \sum_{x_{\mathcal{L}} \neq v_{\mathcal{L}}:  x_{\bar{\mathcal{L}}} = v_{\bar{\mathcal{L}}} }   f(1, v_{\bar{\mathcal{L}}}) \sum_{x_{\bar{\mathcal{S}}}}  P_{\bold{X}_{\mathcal\mathcal\mathcal\mathcal\mathcal\mathcal\mathcal{S}}}(x_{\mathcal{S}}) \nonumber \\
        & \hspace{0.2in} +   \sum_{x_{\mathcal{L}} \neq v: x_{\bar{\mathcal{L}}} \neq v_{\bar{\mathcal{L}}} } f(1, x_{\bar{\mathcal{L}}})  \sum_{x_{\bar{\mathcal{S}}}}  P_{\bold{X}_{\mathcal\mathcal\mathcal\mathcal\mathcal\mathcal\mathcal{S}}}(x_{\mathcal\mathcal\mathcal\mathcal\mathcal\mathcal\mathcal{S}}) \nonumber \\ 
        & \overset{(c)} = f(1, v_{\bar{\mathcal{L}}}) \operatorname{Pr}(\bold{X}_{\mathcal{L} \cap \mathcal{S}} \neq v_{\mathcal{L} \cap \mathcal{S}}) + \sum_{x_{\bar{\mathcal{L}}} \neq v_{\bar{\mathcal{L}}}} f(1, x_{\bar{\mathcal{L}}}),
        \label{eqn:error_term2}
\end{align}
where step (a) follows from \eqref{eqn:functional_form}, while in step (b), the function $f(0, v_{\bar{\mathcal{L}}}) $ does not depend on $x_{{\mathcal{S}}}$. Step (c) follows from the following: 
\begin{align}
     & \sum_{x_{\mathcal{L}} \neq v_{\mathcal{L}}: x_{\bar{\mathcal{L}}} \neq v_{\bar{\mathcal{L}}} } f(1, x_{\bar{\mathcal{L}}})  \sum_{x_{\bar{\mathcal{S}}}}  P_{\bold{X}_{\mathcal\mathcal\mathcal\mathcal\mathcal\mathcal\mathcal{S}}}(x_{\mathcal\mathcal\mathcal\mathcal\mathcal\mathcal\mathcal{S}})  \nonumber \\ 
     & = \sum_{x_{\mathcal{L}} \neq v_{\mathcal{L}}: x_{\bar{\mathcal{L}}} \neq v_{\bar{\mathcal{L}}} } f(1, x_{\bar{\mathcal{L}}})  \sum_{x_{\bar{\mathcal{S}}}}  P_{\bold{X}_{\bar{\mathcal{S}}}, \bold{X}_{\mathcal{L} \cap \mathcal{S}}}(x_{\bar{\mathcal{S}}}, x_{\mathcal{L} \cap \mathcal{S}}) \nonumber \\ 
          & = \sum_{x_{\mathcal{L}} \neq v_{\mathcal{L}}: x_{\bar{\mathcal{L}}} \neq v_{\bar{\mathcal{L}}} } f(1, x_{\bar{\mathcal{L}}})    P_{\bold{X}_{\mathcal{L} \cap \mathcal{S}}}( x_{\mathcal{L} \cap \mathcal{S}}) \nonumber \\
     & = \sum_{ x_{\bar{\mathcal{L}}} \neq v_{\bar{\mathcal{L}}} } f(1, x_{\bar{\mathcal{L}}})  \sum_{x_{\mathcal{L} \cap \mathcal{S}}}  P_{ \bold{X}_{\mathcal{L} \cap \mathcal{S}}}( x_{\mathcal{L} \cap \mathcal{S}}) = \sum_{ x_{\bar{\mathcal{L}}} \neq v_{\bar{\mathcal{L}}} } f(1, x_{\bar{\mathcal{L}}}).
\end{align}

\subsection{Error Probability for Mechanism $\mathcal{M}_{2}$:}

We write the error probabilities in terms of $f(0, x_{\bar{\mathcal{L}}})$ and following similar steps as in $\mathcal{M}_1$. For the first term, we have the following set of steps:
\begin{align}
    P_{Y, X_{\mathcal{L}}}(0, v_{\mathcal{L}}) & = \sum_{x_{\bar{\mathcal{S}}}} P_{Y| X_{\mathcal{L}}, X_{\bar{\mathcal{S}}}}(0| v_{\mathcal{L}}, x_{\bar{\mathcal{S}}} ) P_{X_{\mathcal{L}}, X_{\bar{\mathcal{S}}}}(v_{\mathcal{L}}, x_{\bar{\mathcal{S}}}) \nonumber \\ 
    & = \sum_{x_{\bar{\mathcal{S}}}} P_{Y| X_{\bar{\mathcal{L}}}, X_{\mathcal\mathcal\mathcal\mathcal\mathcal\mathcal\mathcal{S}}}(0| v_{\bar{\mathcal{L}}}, x_{\mathcal\mathcal\mathcal\mathcal\mathcal\mathcal\mathcal{S}} ) P_{X_{\bar{\mathcal{L}}}, X_{\mathcal\mathcal\mathcal\mathcal\mathcal\mathcal\mathcal{S}}}(x_{\bar{\mathcal{L}}}, x_{\mathcal\mathcal\mathcal\mathcal\mathcal\mathcal\mathcal{S}}) \nonumber \\ & = \sum_{x_{\bar{\mathcal{S}}}} f(0, v_{\bar{\mathcal{L}}}) P_{X_{\mathcal{S}}}(x_{\mathcal{S}})   =  f(0, v_{\bar{\mathcal{L}}}) P_{ X_{\mathcal{L} \cap \mathcal{S}}}(v_{\mathcal{L} \cap \mathcal{S}}). \nonumber 
\end{align}

For the second term, we have the following set of steps: 
\begin{align}
    & \sum_{x_{\mathcal{L}} \neq v_{\mathcal{L}}}  P_{Y, \bold{X}_{\mathcal{L}}}(1,  x_{\mathcal{L}} ) \nonumber \\
    & = \sum_{x_{\mathcal{L}} \neq v_{\mathcal{L}}}  \sum_{x_{\bar{\mathcal{S}}}}  P_{Y|\bold{X}_{\mathcal{L}}, \bold{X}_{\bar{\mathcal{S}}}}(1|  x_{\mathcal{L}}, x_{\bar{\mathcal{S}}} ) P_{\bold{X}_{\mathcal{L}}, X_{\bar{\mathcal{S}}}}(x_{\mathcal{L}}, x_{\bar{\mathcal{S}}}) \nonumber\\ 
     & = \sum_{x_{\mathcal{L}} \neq v_{\mathcal{L}}}  \sum_{x_{\bar{\mathcal{S}}}}  (1-P_{Y| \bold{X}_{\mathcal{L}}, \bold{X}_{\bar{\mathcal{S}}}}(0|  x_{\mathcal{L}}, x_{\bar{\mathcal{S}}} )) P_{\bold{X}_{\mathcal{L}}, \bold{X}_{\bar{\mathcal{S}}}}(x_{\mathcal{L}}, x_{\bar{\mathcal{S}}}) \nonumber\\
              & = 1- P_{\bold{X}_{\mathcal{L}}}( v_{\mathcal{L}})  - \sum_{x_{\mathcal{L}} \neq v_{\mathcal{L}}}  \sum_{x_{\bar{\mathcal{S}}}}  f(0, x_{\bar{\mathcal{L}}}) P_{\bold{X}_{\mathcal{S}}}(x_{\mathcal{S}})\nonumber\\ 
                &  = 1- P_{\bold{X}_{\mathcal{L}}}(v_{\mathcal{L}}) \nonumber \\  
                & \hspace{0.2in}- f(0, v_{\bar{\mathcal{L}}}) \operatorname{Pr}(\bold{X}_{\mathcal{L} \cap \mathcal{S}} \neq v_{\mathcal{L} \cap \mathcal{S}}) - \sum_{x_{\bar{\mathcal{L}}} \neq v_{\bar{\mathcal{L}}}} f(0, x_{\bar{\mathcal{L}}}). \nonumber 
\end{align}

We optimize the per-user error probability for the two release mechanisms for fixed $q$, $s$ and $v_{\mathcal{L}}$ as follows. 

For $\mathcal{M}_1$: from \eqref{eqn:error_term1} and \eqref{eqn:error_term2}, the per-user error probability is 
\begin{align}
    P_{e,1}
    & =P_{\bold{X}_{\mathcal{L}}}(v_{\mathcal{L}}) + \sum_{x_{\bar{\mathcal{L}}} \neq v_{\bar{\mathcal{L}}}} f(1, x_{\bar{\mathcal{L}}})  \nonumber\\
    & \hspace{0.2in} + f(1, v_{\bar{\mathcal{L}}}) (2 \operatorname{Pr}(\bold{X}_{\mathcal{L} \cap \mathcal{S}} \neq v_{\mathcal{L} \cap \mathcal{S}}) - 1 ).
\end{align}
We minimize the probability of error as follows. Here, we have two cases depending on the value of $\operatorname{Pr}(\bold{X}_{\mathcal{L} \cap \mathcal{S}} \neq v_{\mathcal{L} \cap \mathcal{S}}) $. If  $\operatorname{Pr}(\bold{X}_{\mathcal{L} \cap \mathcal{S}} \neq v_{\mathcal{L} \cap \mathcal{S}}) \geq 1/2$, we pick $f(1, x_{\bar{\mathcal{L}}}) = 0, \forall x_{\bar{\mathcal{L}}}$. When $\operatorname{Pr}(\bold{X}_{\mathcal{L} \cap \mathcal{S}} \neq v_{\mathcal{L} \cap \mathcal{S}}) < 1/2$, we pick $f(1, v_{\bar{\mathcal{L}}}) = \min_{w} P_{\bold{X}_{\bar{\mathcal{L}}}|\bold{X}_{\mathcal{S}}}(x_{\bar{\mathcal{L}}} | x_{\mathcal{S}} = w)$ and    $f(1, x_{\bar{\mathcal{L}}} ) = 0, \forall x_{\bar{\mathcal{L}}} \neq v_{\bar{\mathcal{L}}}$.

For $\mathcal{M}_2$ The per-user error probability for $\mathcal{M}_2$ is
\begin{align}
    P_{e,2} &= 1- P_{\bold{X}_{\mathcal{L}}}(v_{\mathcal{L}}) - \sum_{x_{\bar{\mathcal{L}}} \neq v_{\bar{\mathcal{L}}}} f(0, x_{\bar{\mathcal{L}}}) \nonumber \\
   & \hspace{0.2in} - f(0, v_{\bar{\mathcal{L}}}) ( 2 \operatorname{Pr}(\bold{X}_{\mathcal{L} \cap \mathcal{S}} \neq v_{\mathcal{L} \cap \mathcal{S}}) - 1),
\end{align}
where the error probability is minimized when $f(0, x_{\bar{\mathcal{L}}}) =  \min_{w} P_{\bold{X}_{\bar{\mathcal{L}}}|X_{\mathcal{S}}}(x_{\bar{\mathcal{L}}} | w), \forall x_{\bar{\mathcal{L}}} \neq v_{\bar{\mathcal{L}}} $. For $f(0, v_{\bar{\mathcal{L}}})$, we have two special cases: when $\operatorname{Pr}(X_{\mathcal{L} \cap \mathcal{S}} \neq v_{\mathcal{L} \cap \mathcal{S}})> 1/2$, we pick $f(0, v_{\bar{\mathcal{L}}}) =  \min_{w} P_{\bold{X}_{\bar{\mathcal{L}}}|\bold{X}_{\mathcal{S}}}(v_{\bar{\mathcal{L}}} | w) $, otherwise we set $f(0, v_{\bar{\mathcal{L}}}) = 0$.

This concludes the proof for Theorem 2.

\section{Proof of Theorem \ref{thm:lowerbound}}
\begin{proof}
According to Fano's inequality\cite{doi:https://doi.org/10.1002/047174882X.ch17}, the lower bound of the error probability can be expressed as:
\begin{equation}\label{eq1}
    h(A|Y)\le{h(P_e)+P(e)\log(|A|-1)}.
\end{equation}
As $|A|=2$, Eq. \eqref{eq1} can be expressed as:
\begin{equation}\label{eq2}
\begin{aligned}
      h(P_e)\ge{h(A|Y)}=h(A)-I(A;Y).
\end{aligned}
\end{equation}
In Eq.\eqref{eq2}, $h(A)$ is a constant given the distribution of $A$ (which is a deterministic function of $\bold{X}_{\mathcal{L}}$). Next, we focus on $I(A;Y)$. To derive a lower bound of $P_e$, we are looking at the upper bound of $I(A;Y)$. %Which can be derived by the following steps:
Note that random variable $A$ can be represented as:
\begin{equation}
\begin{aligned}
    A=\mathbbm{1}_{\{\bold{X}_{\mathcal{L}}=v_{\mathcal{L}}\}}
    =\mathbbm{1}_{\{\bold{X}_{\bar{\mathcal{L}}}=v_{\bar{\mathcal{L}}}\}}\times \mathbbm{1}_{\{\bold{X}_{\mathcal{L}\cap\mathcal{S}}=v_{\mathcal{L}\cap\mathcal{S}}\}}.
\end{aligned}
\end{equation}
Denote $A_{\bar{\mathcal{L}}}=\mathbbm{1}_{\{\bold{X}_{\bar{\mathcal{L}}}=v_{\bar{\mathcal{L}}}\}}$ and $A_{{\mathcal{L}\cap\mathcal{S}}}=\mathbbm{1}_{\{\bold{X}_{\mathcal{L}\cap\mathcal{S}}=v_{\mathcal{L}\cap\mathcal{S}}\}}$. Therefore, using this notation, we can write 
 $A=A_{\bar{\mathcal{L}}}\times A_{{\mathcal{L}\cap\mathcal{S}}}$.
Then the upper bound of $I(A;Y)$ can be derived as:

\begin{align}\label{eqx}
    %\begin{aligned}
    I(A;Y)=&I(A;Y|\bold{X}_{\mathcal{S}})\\
    \le&{I(A;\bold{X}|\bold{X}_{{\mathcal{S}}})}\\
    =&h(A|\bold{X}_{{\mathcal{S}}})-h(A|\bold{X},\bold{X}_{{\mathcal{S}}})\nonumber\\
    =&h(A|\bold{X}_{{\mathcal{S}}})\nonumber,
    %\end{aligned}
\end{align}
where (22) follows from the privacy constraint in (1), (23) follows by data processing. From another perspective:
\begin{align}\label{eqy}
    I(A;Y)=&I(A_{\bar{\mathcal{L}}}A_{\mathcal{L}\cap\mathcal{S}};Y)\nonumber\\
    \le&I(A_{\bar{\mathcal{L}}},A_{\mathcal{L}\cap\mathcal{S}};Y)\nonumber\\
    =&I(A_{{\mathcal{L}\cap\mathcal{S}}};Y)+I(A_{\bar{\mathcal{L}}};Y|A_{{\mathcal{L}\cap\mathcal{S}}})\nonumber\\
    =&I(A_{\bar{\mathcal{L}}};Y|A_{{\mathcal{L}\cap\mathcal{S}}})\\
    \le&{I(A_{\bar{\mathcal{L}}};X|A_{{\mathcal{L}\cap\mathcal{S}}})}\\
    =&h(A_{\bar{\mathcal{L}}}|A_{{\mathcal{L}\cap\mathcal{S}}})-h(A_{\bar{\mathcal{L}}}|X,A_{{\mathcal{L}\cap\mathcal{S}}})\nonumber\\
    =&h(A_{\bar{\mathcal{L}}}|A_{{\mathcal{L}\cap\mathcal{S}}}),\nonumber
\end{align}
where (24) follows from the privacy constraint in (1), (25) follows from data processing. From \eqref{eqx} and \eqref{eqy}, we obtain an upper bound on $I(Y;A)$ as follows:
$I(A;Y) \le \min (h(A_{\bar{\mathcal{L}}}|A_{{\mathcal{L}\cap\mathcal{S}}}),h(A|X_{{\mathcal{S}}}))$
Then, substituting the above bound in \eqref{eq2}, we arrive at the lower bound on error probability stated in Theorem 3.  
\begin{equation}
    h(P_e)\ge{h(A)-\min\{h(A_{\bar{\mathcal{L}}}|A_{{\mathcal{L}\cap\mathcal{S}}}),h(A|X_{{\mathcal{S}}})\}}.
\end{equation}
\end{proof}
\section{Proof of Claims made in Remark 2}

1) when $\mathcal{L}\cap{\mathcal{S}}=\emptyset$, from Theorem \ref{thm:lowerbound}, we have:
\begin{equation}\label{eq8}
    P_e\ge h^{-1}\big(I(A;\bold{X}_{\mathcal{S}})\big).
\end{equation}

On the other hand, according to Theorem \ref{thm:Pe}, when $\mathcal{L}\cap{\mathcal{S}}=\emptyset$, $\mathcal{E}=0$, and 
\begin{equation}\label{eq9}
\begin{aligned}
    P_{e,1}=&P_{\bold{X}_{\mathcal{L}}}(v_{\mathcal{L}})  - \min_{w} P_{\bold{X}_{{\mathcal{L}}}|\bold{X}_{\mathcal{S}}}(v_{{\mathcal{L}}}|  w),\\
    P_{e,2} =&1 - P_{\bold{X}_{\mathcal{L}}}(v_{\mathcal{L}}) -  \sum_{x_{{\mathcal{L}}} \neq v_{{\mathcal{L}}} } \min_{w} P_{\bold{X}_{{\mathcal{L}}}|\bold{X}_{\mathcal{S}}}(x_{{\mathcal{L}}}| w).
\end{aligned}
\end{equation}
We next consider two sub-cases regarding the relationship between $A$ and $X_{\mathcal{S}}$:

a). When $I(A;\bold{X}_{\mathcal{S}})=h(A)$, from \eqref{eq8}:
\begin{equation}
    \begin{aligned}
        P_e\ge &h^{-1}(h(A)-H(A|\bold{X}_{\mathcal{S}}))
        =h^{-1}(h(A))\\
        =&\min\{P_A(0),P_A(1)\}.
    \end{aligned}
\end{equation}
Since $h(A|X_{\mathcal{S}})=0$, $\min_{w} P_{\bold{X}_{{\mathcal{L}}}|\bold{X}_{\mathcal{S}}}(x_{{\mathcal{L}}}|  w)=0$ for all $x_{\mathcal{L}}$. From \eqref{eq9}, we have $P_{e,1}=P_A(1)$, $P_{e,2}=P_A(0)$, which implies that the minimal $P_e$ resulted by $\mathcal{M}_1$ and $\mathcal{M}_2$ is $\min\{P_A(0),P_A(1)\}$.

b) When $I(A;\bold{X}_{\mathcal{S}})=0$, from \eqref{eq8}, the lower bound implies $P_e\ge{0}$.
From \eqref{eq9}, we have $P_{e,1}=P_{e,2}=0$, which implies that the minimal $P_e$ resulted by $\mathcal{M}_1$ and $\mathcal{M}_2$ is $0$.

The above discussion shows that under the two extreme cases, the proposed mechanism achieves $P_e$ matches the lower bound.

%Consider the following two special cases: $R(X_{\mathcal{L}},X_{\mathcal{S}})=0$ (strong correlation), and $R(X_{\mathcal{L}},X_{\mathcal{S}})=1$ (weak correlation).

2) When $\mathcal{L}\in\mathcal{S}$, $H(A|\bold{X}_{{\mathcal{S}}})=H(A_{\bar{\mathcal{L}}}|A_{{\mathcal{L}\cap\mathcal{S}}})=0$. The result follows sub-case a), which matches the lower bound from Theorem \ref{thm:lowerbound}. The intuition is when all queried locations are sensitive, both mechanisms sample random answers according to the data prior to keep the queried sequence perfectly private.

\section{Proof of proposition 1}
We next show the lower bound of the $EAE$ under the proposed mechanisms, the goal is to show that when users are i.i.d, the bounds are tight (i.e. the inequality can be replaced by equality).

In the following, we let $Y$ denote the aggregated results from each local answer $Y^{(k)}$: $Y=\sum_{k=1}^KY^{(k)}$, and $A$ denote the real answer from each user: $A=\sum_{k=1}^KA^{(k)}$.
By Jensen's Inequality:
\begin{small}
\begin{equation}
\begin{aligned}
    &E|Y-A|
    \ge{|E(Y-A)|}\\
    =&\left|E\sum_{k=1}^K(Y^{(k)}-A^{(k)})\right|\\
    =&\left|\sum_{k=1}^KE(Y^{(k)}-A^{(k)})\right|\\
    =&\left|\sum_{k=1}^K\left[P_{Y^{(k)},A^{(k)}}(1,0)-P_{Y^{(k)},A^{(k)}}(0,1)\right]\right|\\
    =&\left|\sum_{k=1}^KP_{_{Y^{(k)}|A^{(k)}}}(1|0)P_{A^{k}}(0)-\sum_{k=1}^KP_{Y^{(k)}|A^{(k)}}(0|1)P_{A^{(k)}}(1)\right|\\
    =&\big|\sum_{k,x_{\bar{\mathcal{S}}}}P(Y^{(k)}=1|\bold{X}_{\mathcal{L}}\neq{v_{\mathcal{L}}},\bold{X}_{\bar{\mathcal{S}}}=x_{\bar{\mathcal{S}}})\\
    &P(\bold{X}_{\mathcal{S}}=x_{\mathcal{S}}|\bold{X}_{\mathcal{L}}\neq{v_{\mathcal{L}}})P(\bold{X}_{\mathcal{L}}\neq{v_{\mathcal{L}}})\\
    -&\sum_{k,x_{\bar{\mathcal{S}}}}P(Y^{(k)}=0|\bold{X}_{\mathcal{L}}={V_{\mathcal{L}}},\bold{X}_{\bar{\mathcal{S}}}=x_{\bar{\mathcal{S}}})\\
    &P(\bold{X}_{\mathcal{S}}=x_{\mathcal{S}}|\bold{X}_{\mathcal{L}}={v_{\mathcal{L}}})P(\bold{X}_{\mathcal{L}}={v_{\mathcal{L}}})\big|,
\end{aligned}
\end{equation}
\end{small}
where $P(Y^{(k)}=0|\bold{X}_{\mathcal{L}}={v_{\mathcal{L}}},\bold{X}_{\bar{\mathcal{S}}}=x_{\bar{\mathcal{S}}})$ and $P(Y^{(k)}=1|\bold{X}_{\mathcal{L}}\neq{v_{\mathcal{L}}},\bold{X}_{\bar{\mathcal{S}}}=x_{\bar{\mathcal{S}}})$ are parameters of the mechanisms. If users are i.i.d, substituting $\mu_1$ we have: 
\begin{equation}
    |E(Y-A)|=K\left|\left[P_{\bold{X}_{\mathcal{L}}}(v_{\mathcal{L}})-\min_{w}P_{\bold{X}_{\bar{\mathcal{L}}}|\bold{X}_{\mathcal{S}}}(v_{\bar{\mathcal{L}}|w})P_{\bold{X}_{\mathcal{L}\cap{\mathcal{S}}}}(v_{\mathcal{L}\cap{\mathcal{S}}})\right]\right|.
\end{equation}
Similarly, substituting $\mu_2$, we have: 
\begin{equation}
    |E(Y-A)|=K\left|\left[1-P_{\bold{X}_{\mathcal{L}}}(v_{\mathcal{L}})-\sum_{u_{\bar{\mathcal{L}}}\neq{v_{\bar{\mathcal{L}}}}}\min_{w}P_{\bold{X}_{\bar{\mathcal{L}}}|\bold{X}_{\mathcal{S}}}(u_{\bar{\mathcal{L}}|w})\right]\right|.
\end{equation}
As such, when users are i.i.d (same distribution leads to the same mechanism, either $\mathcal{M}_1$ or $\mathcal{M}_2$), the lower bound and the upper bound of the proposed mechanisms match each other, which means under the i.i.d assumption, the bounds we derived are tight.

\section{Proof of Theorem 4}
\begin{proof}
\begin{equation}
\begin{aligned}
    &P_{Y|\bold{X}_{\mathcal{S}}}(y|\bold{x}_{\mathcal{S}})\\ = &\sum_{a}P_{Y|A, \bold{X}_{\mathcal{S}}}(y|a,\bold{x}_{\mathcal{S}})P_{A|\bold{X}_{\mathcal{S}}}(a|\bold{x}_{\mathcal{S}})\\
    =&\sum_{a\neq{y}}P_{Y|A, \bold{X}_{\mathcal{S}}}(y|a,\bold{x}_{\mathcal{S}})P_{A|\bold{X}_{\mathcal{S}}}(a|\bold{x}_{\mathcal{S}})\\ + &P_{Y|A, \bold{X}_{\mathcal{S}}}(y|y,\bold{x}_{\mathcal{S}})P_{A|\bold{X}_{\mathcal{S}}}(y|\bold{x}_{\mathcal{S}}),\\
    =&\sum_{a\neq{y}}P_A(y)\left[1-\frac{\min_wP_{A|\bold{X}_{\mathcal{S}}}(a|w)}{P_{A|\bold{X}_{\mathcal{S}}}(a|\bold{x}_{\mathcal{S}})}\right]P_{A|\bold{X}_{\mathcal{S}}}(a|\bold{x}_{\mathcal{S}})\\ + &\left[P_A(y)+(1-P_A(y))\frac{\min_wP_{A|\bold{X}_{\mathcal{S}}}(a|w)}{P_{A|\bold{X}_{\mathcal{S}}}(a|\bold{x}_{\mathcal{S}})}\right]P_{A|\bold{X}_{\mathcal{S}}}(y|\bold{x}_{\mathcal{S}})\\
    =&P_A(y)\sum_{a\neq{y}}\left[P_{A|\bold{X}_{\mathcal{S}}}(a|\bold{x}_{\mathcal{S}})-{\min_wP_{A|\bold{X}_{\mathcal{S}}}(a|w)}\right]\\ + &P_A(y)P_{A|\bold{X}_{\mathcal{S}}}(y|\bold{x}_{\mathcal{S}})+(1-P_A(y)){\min_wP_{A|\bold{X}_{\mathcal{S}}}(a|w)}\\
    \end{aligned}
\end{equation}
which can be further expressed as:
\begin{equation*}
\begin{aligned}
    % &1+P_A(y)\sum_{a}P_{A|\bold{X}_{\mathcal{S}}}(a|\bold{x}_{\mathcal{S}})-P_A(y)\sum_{a}\min_wP_{A|\bold{X}_{\mathcal{S}}}(a|w)\\
    1+P_A(y)-P_A(y)\sum_{a}\min_wP_{A|\bold{X}_{\mathcal{S}}}(a|w).\\
\end{aligned}
\end{equation*}
As a result, $P_{Y|\bold{X}_{\mathcal{S}}}(y|\bold{x}_{\mathcal{S}})$ is independent of $\bold{X}_{\mathcal{S}}$.
\end{proof}

\section{Proof of Theorem 5}
\begin{proof}
\begin{equation*}
    \begin{aligned}
        & E[|Y-A|]
        % =&\sum_{y}\sum_{a\neq{y}}|y-a|P_{Y,A}(y,a)\\
        =\sum_{y}\sum_{a\neq{y}}\sum_{\bold{x}_{\mathcal{S}}}|y-a|P_{Y,A,\bold{X}_{\mathcal{S}}}(y,a,\bold{x}_{\mathcal{S}})\\
        =&\sum_{y}\sum_{a\neq{y}}\sum_{\bold{x}_{\mathcal{S}}}|y-a|P_{Y|A,\bold{X}_{\mathcal{S}}}(y|a,\bold{x}_{\mathcal{S}})P_{A|\bold{X}_{\mathcal{S}}}(a|\bold{x}_{\mathcal{S}})P_{\bold{X}_{\mathcal{S}}}(\bold{x}_{\mathcal{S}})\\
        =&\sum_{y}\sum_{a\neq{y}}\sum_{\bold{x}_{\mathcal{S}}}|y-a|P_{\bold{X}_{\mathcal{S}}}(\bold{x}_{\mathcal{S}})\\
        &\cdot P_A(y)\left[1-\frac{\min_wP_{A|\bold{X}_{\mathcal{S}}}(a|w)}{P_{A|\bold{X}_{\mathcal{S}}}(a|\bold{x}_{\mathcal{S}})}\right]P_{A|\bold{X}_{\mathcal{S}}}(a|\bold{x}_{\mathcal{S}})\\
        =&\sum_{y}\sum_{a\neq{y}}\sum_{\bold{x}_{\mathcal{S}}}|y-a|\\
        &\cdot P_A(y)\left[P_{A,\bold{X}_{\mathcal{S}}}(a,\bold{x}_{\mathcal{S}})-{\min_wP_{A|\bold{X}_{\mathcal{S}}}(a|w)P_{\bold{X}_{\mathcal{S}}}(\bold{x}_{\mathcal{S}})}\right]\\
        =&\sum_{a=0}^N\sum_{y\neq{a}}|y-a|P_A(y)\left[P_A(a)-\min_{w}P_{A|X_{\mathcal{S}}}(a|w)\right].
    \end{aligned}
\end{equation*}
\end{proof}
 \bibliographystyle{IEEEtran}
\bibliography{ref}

\end{document}